\documentclass[]{article}

\pdfoutput=1
\usepackage{setspace} 

\catcode`\@=11\relax
\usepackage{fancyhdr,timestamp}
\usepackage{amsthm,amssymb}
\usepackage{latexsym}

\usepackage{algorithm,algorithmic}

\usepackage{rotating} 

\def\FP#1{{\bf ({#1})}}
\def\CP#1{{\em {#1}}}

\newtheorem{thm}{Result}

\newtheorem{prop}[thm]{Proposition}

\newtheorem{hypothesis}{Hypothesis}

\def\reftab#1{Table~\ref{#1}}

\def\refresult#1{Result~\ref{#1}}

\def\etal{{\em et~al.}}
\def\eps{\ensuremath{\varepsilon}}
\def\Reals{\ensuremath{\mathbb{R}}}
\def\grad{\ensuremath{\nabla}}
\def\del{\ensuremath{\partial}}
\def\norm#1{\ensuremath{\left\|{#1}\right\|}}

\newlength{\bigfigwidth}
\newlength{\figwidth}
\def\spcing{\vspace{1ex}}

\newtheorem{mech}{Schema}
\newtheorem{theorem}[thm]{Theorem}
\newtheorem{lemma}[thm]{Lemma}
\newtheorem{corol}[thm]{Corollary}

\theoremstyle{definition}
\newtheorem{definition}[thm]{Definition}

\theoremstyle{remark}
\newtheorem{remark}[thm]{Remark}
\def\inner#1#2{\ensuremath{\left< {#1} , {#2}\right>}}
\def\bsupp#1#2{\ensuremath{ \mathrm{bsupp}({#1} ; {\del#2}) }}

\def\paren#1{\ensuremath{\left( {#1} \right)}}
\def\st#1{\ensuremath{\left\{ {#1} \right\}}}
\def\abs#1{\ensuremath{\left| {#1} \right|}}
\def\norm#1{\ensuremath{\left\| {#1} \right\|}}
\def\del{\ensuremath{\partial}}
\def\Reals{\ensuremath{\mathbb{R}}}
\def\shapesym{\ensuremath{\Omega}}
\def\distsym{\ensuremath{\mathcal{E}}}
\def\MA{\ensuremath{\mathbf{MA}}}
\def\approxct{\ensuremath{\tilde c_t}}

\def\grad{\ensuremath{\nabla}}
\def\eps{\ensuremath{\varepsilon}}

\newcommand{\reffig}[2][{}]{Fig.~\ref{#2}{\em #1}}

\def\refsec#1{Section~\ref{#1}}

\def\refeq#1{Eq.~\ref{#1}}
\def\refthm#1{Theorem~\ref{#1}}
\def\reflemma#1{Lemma~\ref{#1}}
\def\refprop#1{Proposition~\ref{#1}}
\def\reflem#1{Lemma~\ref{#1}}

\def\c1D{\ensuremath{ {c_{1D}}} }

\def\partderiv#1#2{\ensuremath{ \frac{\del {#1}}{\del {#2}} } }
\def\conv{\ensuremath{\star}}

\def\d#1{\ensuremath{\mathrm{d}{#1}}}
\def\v#1{\ensuremath{\mathbf{#1}}}

\def\defem#1{{\em {#1}}}
\def\dt{\ensuremath{\Delta t}}

\usepackage{graphicx}
\graphicspath{{figs/}}

\title{ {\bf Distance Maps and Plant Development \#1: Uniform
    Production and Proportional Destruction} } 

\author{Pavel Dimitrov and Steven W. Zucker} 

\date{\today}

\begin{document}

\maketitle

\begin{abstract}

Experimental data regarding auxin and venation formation exist at both
macroscopic and molecular scales, and we attempt to unify them into a
comprehensive model for venation formation.  We begin with a set of
principles to guide an abstract model of venation formation, from
which we show how patterns in plant development are related to the
representation of global distance information locally as
cellular-level signals.  Venation formation, in particular, is a
function of distances between cells and their locations.  The first
principle, that auxin is produced at a constant rate in all cells,
leads to a (Poisson) reaction-diffusion equation. Equilibrium
solutions uniquely codify information about distances, thereby
providing cells with the signal to begin differentiation from ground
to vascular.  A uniform destruction hypothesis and scaling by cell
size leads to a more biologically-relevant (Helmholtz) model, and
simulations demonstrate its capability to predict leaf and root auxin
distributions and venation patterns.  The mathematical development is
centered on properties of the distance map, and provides a mechanism
by which global information about shape can be presented locally to
individual cells.  The principles provide the foundation for an
elaboration of these models in a companion paper \cite{plos-paper2},
and together they provide a framework for understanding organ- and
plant-scale organization.
\end{abstract}

\tableofcontents

\newpage

\section{Introduction}

One of the principal tenants of biology is that no matter how large an
organism becomes everything about it must ultimately have an
explanation at the cellular level. Molecular biology goes even further
by requiring an explanation on the level of chemical reactions. What
chemical compounds and what reactions give rise to the intricate
patterns of veins in leaves? Or to the pattern of specialized cells in
the root of a plant? These are the types of questions that modern
biology attempts to answer. And these same questions have prompted
workers from the other sciences to join in. Physicists, mathematicians
and computer scientists find the problems especially intriguing
because of the need to explain how global patterns develop from local
behaviors. Seen in this light, the problem becomes the search for a
``plant geometry:'' how cells determine where they are located with
respect to other ``special'' cells, what those ``special'' cells are,
and how cells behave once the information becomes available.

Attempts to solve a version this problem, in which the notion of
positional information is the focus, can be traced back to over a
century ago \cite{voechting-1877, voechting-1878}, but it was only
within the last fifty years that mechanistic proposals were first
submitted \cite{wolpert-1969, wolpert-1981, gierer72:initial}. Of
these, the idea of a diffusible morphogen
\cite{turing52:morphogenesis} has received the most attention because
it captures complex measurable phenomena in a compact mathematical
form. This so-called reaction-diffusion formulation requires specific
knowledge of molecular interactions, but its typical form requires at
least two chemical substances in order to explain a patterning
phenomenon \cite{harrisonbook}. By contrast, we have shown
\cite{pdswz:paper} that a modification of the original formulation
only requires one substance and already provides preliminary answers
to the three main questions of plant geometry.

This paper is the first of three, in which we develop these questions
in further detail. The series looks for the simplest hypotheses that
can explain patterning phenomena arising in vein formation,
facilitated transport of plant hormones, cell division and expansion,
hormone concentration distributions, and others. We propose hypotheses
about the local behavior of cells, such as how a hormone is produced,
and analyze them mathematically to explain observed phenomena as well
as to generate further hypotheses. Thus, some of our fundamental
assumptions will be theoretically derived. The verification of their
implications is presented through numerical simulations, which we
demonstrate to afford unique interpretations. As a result, we develop
a theory that explains how discrete systems -- such as collections of
cells -- may compute a distance map in a variety of scenarios and
represented by various interpretations of hormone concentrations.

We begin, in this paper, by refining some of our assumptions in
\cite{pdswz:paper} about the biology of plants. To keep matters
tractable, it is still necessary to include some abstraction for what
otherwise might be considered signalling or other networks. We
abstract ``constant production'', ``proportional destruction'', and
c-vascular conversion ``schema'' in this paper. We provide the
mathematical analysis of our earlier model, to which we refer as the
Poisson Model, that proves our claims in \cite{pdswz:paper}. Then, we
extend the model by introducing a more biologically plausible
assumption about the destruction of the signaling hormone auxin, which
gives rise to the Helmholtz Model. Mathematical analysis of this
formulation demonstrates that the properties of the Poisson model are
kept and that it can explain even more experimental observations.

In the next paper \cite{plos-paper2} we elaborate the schema into more
biologically plausible mechanisms using different transport
facilitators and the chemosmotic theory.  We observe that it is
remarkable that the Fickian transport and reaction diffusion equation
developed here can provide this abstraction in such a manner that its
main properties hold when more detailed facilitated transport is taken
into account.  Our goal throughout this series of papers
\cite{pdswz:paper, plos-paper1, plos-paper2, plos-paper3} is to
formulate those abstract principles that can govern the qualitative
properties exhibited at the systems level in plants. Such an approach
is necessary, we believe, to organize into organ- and plant-scale
syntheses the diversity of cellular and molecular mechanisms
constantly being discovered. As we show through the series, an
elaboration of the principles into increasing detail predicts
non-linear and, at times, surprisingly delicate sequential
developmental patterns. Without such a systems-level understanding one
might be tempted to postulate a need for unnecessary genetic
machinery.







\section{Constant Production Hypothesis: Poisson Model}

In \cite{pdswz:paper}, we proposed the Constant Production Hypothesis
and argued that rich geometric information becomes available to cells
in a single-substance reaction-diffusion model. 

\begin{hypothesis}[CPH]
  \label{hypothesis:cph}
  Auxin is produced in all cells at the same constant rate. 
\end{hypothesis}

\noindent
In this section, we recall the main consequence of this assumption and
then prove it mathematically.

\subsection{Background}

A leaf is a collection of cells. We distinguish between {\em ground}
cells, those that give rise to all others, and {\em vascular} cells,
those that comprise the venation pattern. We focus on early leaf
development and concentrate on signals sufficient to initiate the
cascade of events that change ground cells into vascular cells within
an expanding areole. To keep matters tractable, a cell will be
referred to as {\em c-vascular} (cascade vascular) immediately after
this cascade is initiated. The role of this c-vascular abstraction is
to summarize the increasingly elaborate cascade of genetic expression
and transcription regulation that is being uncovered; see
\cite{scarpella:genes-for-procambium}. The sub-collection of
c-vascular cells may be thought of as an early pre-pattern from which
veins derive. Ground cells have (essentially) homogeneous
characteristics and areoles are delimited by more developed c-vascular
(or mature vascular) cells.  Instead of assuming the pre-pattern is
predefined, our model establishes how it emerges from local
operations. We refer to both membranes and cell walls together as {\em
  cell interfaces} and assume that they act as a single membrane.

\subsection{Poisson Model}

\begin{table}[t]
  \centering

\begin{tabular}{cp{1ex}c}
  \begin{tabular}[t]{|p{0.9\textwidth}|}
    \hline
    {\bf Poisson Model}\\
    \hline
    {\bf Definitions}\\
    
    {\em Ground Cell:}  Diffusion coefficient $D_{g}$. \\
    
    {\em C-Vascular Cell:} At least one interface has diff. coef. $D_{v} >
    D_{g}$.\\
    
    \hline
    {\bf Cell Functions} (Program) \\

    {\sc CF1:} Produce substance $s$ at constant rate $K$.
    
    {\sc CF2:} Measure $c$ and $\Delta c$ through interfaces.\\

    {\sc CF3:} Diffuse $s$ through interfaces.\\
    
    {\sc CF4:} When $\Delta c > \tau$ through interface $I$, change its
    diffusion coefficient to $D_v$. \\
    
    \hline
  \end{tabular}

\end{tabular}
  \caption[Poisson model definitions.]{Poisson model definitions. The
    mechanism for changing the diffusion coefficent in CF4 will be
    elaborated in \cite{plos-paper2}.}
  \label{tab:poisson-model}
\end{table}

Each cell performs the basic functions listed in
\reftab{tab:poisson-model} independently and simultaneously.  Under
these assumptions, then, cell functions CF1 and CF3 determine the
equation governing the distribution of $s$ in the areole. They define
how the substance is produced and transported for both ground and
c-vascular cells. The latter evacuate the hormone much faster so we
assume that the boundary of the areole may be thought of as a sink for
$s$, i.e. it is essentially kept at a constant level. Therefore, the
temporal change of the concentration inside a region depends on how
much diffuses in or out of a cell plus how much is created; in
symbols,
\begin{equation}\label{eq:goveq}
c_t = D \grad^2 c + K,
\end{equation}
where $D$ is the diffusion constant of ground cells, $\grad^2 c =
c_{xx} + c_{yy}$ is the Laplacian of concentration over cell position,
and $K$ is as in CF1. This is a reaction-diffusion equation which has
a steady-state: after a sufficiently long time, the dynamical system
is well approximated by the $c$ such that $c_t=0$ (see
\reffig{fig:big-picture}). Observe that those cells which are further
from the boundary have higher concentrations.  In fact, the
concentration profile is qualitatively similar to that of the function
assigning to each cell the shortest distance to a (c-)vascular
cell---the so-called {\em distance
  transform}~\cite{blum73:the-instigator}.

When $c_t = 0$, \refeq{eq:goveq} becomes a standard Poisson equation.
Given our boundary conditions ($c=0$ at veins), there is a unique $c$
satisfying it \cite{pdswz:proof-TR}. From this, we
calculate:

\begin{thm} \label{cph-result}
  Consider an areole and suppose that $P$ is a ground cell which is
  furthest from the c-vascular boundary. Let $Q$ be a c-vascular cell
  which is closest to $P$ and denote by $L$ the distance between $P$
  and $Q$. Then
  \begin{enumerate}
    \item[(a)] $c(P)$ is proportional to $\frac{K}{D} L^2$;
      
    \item[(b)] the change in $c$ at the interface of $Q$ nearest to
      $P$ is proportional to $\frac{K}{D} L$.
      
    \item[(c)] $\Delta c$ is largest at an interface of the c-vascular
      boundary, larger than for any ground cell, and is proportional
      to $\frac{K}{D} L$.
  \end{enumerate}
\end{thm}

Therefore, using part (a) and CF2, a cell may determine if it has
become further than $L$ units from the closest c-vascular (supply)
cell by measuring its concentration. More must be done, however, to
guarantee that the developing vascular network is connected, and
utilizing the difference in concentration accomplishes this.

\refresult{cph-result}(b) asserts that $\Delta c$ at the venation is
directly proportional to $L/D$ and does not depend on the value of
$c$. Moreover, it also gives the direction toward the furthest cell.
This is sufficient to show that mechanisms for new strand creation
should adhere to the following schema:
\begin{mech}\label{mech:vein}
  Let $D_I$ be the diffusion constant across an interface $I$ and
  $\Delta c$ be the concentration difference through $I$. Then
  increase $D_I$ to a higher value when $\Delta c >
  \alpha\frac{K}{D_I} L_0$.  ($\alpha$ is a constant of
  proportionality.) Alternatively, the flux $\phi = D_I \Delta c =
  \alpha K L_0$ may be employed.
\end{mech}

An illustration of this mechanism is shown in
\reffig{fig:big-picture}.

\begin{figure*}[ht]
\centering

\includegraphics[width=\textwidth]{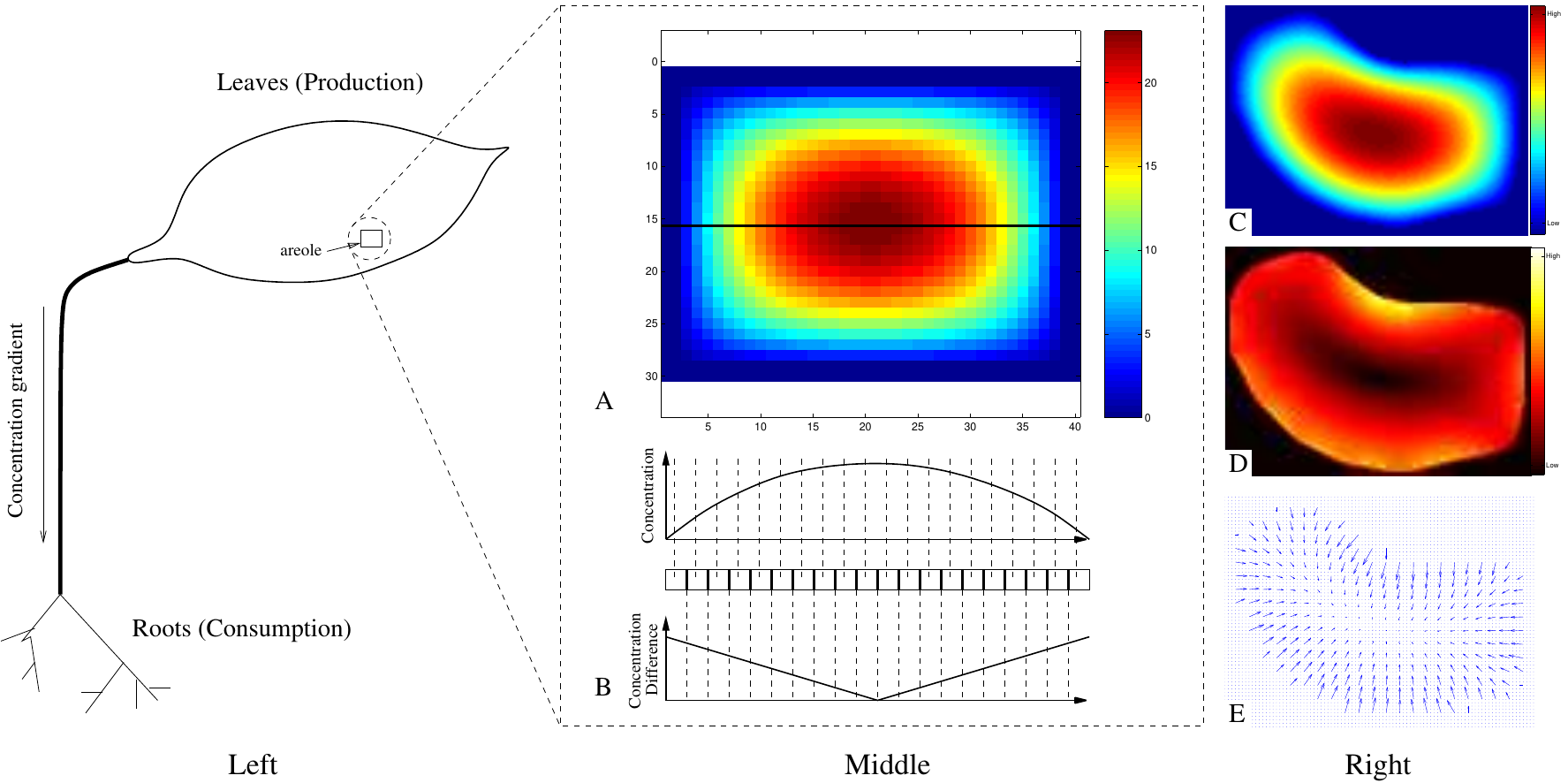} 

\caption[Hormone concentration inside an
  areole.]{\label{fig:big-picture} Hormone concentration inside an
  areole. ({\sc left}) A rectangular domain (artificial areole) is
  illustrated with a boundary of c-vascular cells.  Assuming
  c-vascular cells are much more efficient at transporting $s$, the
  boundary may be taken as a sink and $c$ is governed by
  \refeq{eq:goveq}.  ({\sc middle}) $c$ at near steady-state, $c_t
  \approx 0$. Also shown are the values of $c$ and $\Delta c$ along a
  path (in black) across the areole. Notice how the concentration
  peaks at cells furthest from the veins, while $\Delta c$ peaks near
  the vein. ({\sc right}) (C) Concentration, (D)
  magnitude of gradient, (E) gradient vector field.  Observe how the
  gradient vectors point toward largest concentration increase.}

\end{figure*}

\subsection{Analysis of the Poisson Model}

\subsubsection{Definitions and Background: Geometry}

A collection of ground cells surrounded by c-vascular cells is called
an {\em areole}. Our technical result will assume that an areole is a
discretization of a continuous portion of $\Reals^2$ which we call a
shape.

\begin{definition}
  A {\em shape} is any subset $\shapesym\in\Reals^2$ which is the
  closure of a bounded open set and has a boundary $\del\shapesym$
  consisting of finitely many smooth curves.
  
  A point $Q\in\del\shapesym$ is {\em concave} if for any line $\ell$
  locally tangent to $Q$ there is an open ball $B_\eps(Q)$ such that
  $B_\eps(Q)\cap \ell\cap(\shapesym-\del\shapesym) = B_\eps(Q)\cap
  \ell - \st{Q}$ (i.e. the line segment is inside $\shapesym$). If
  $B_\eps(Q)\cap \ell\cap\shapesym \subset \del\shapesym$, then $Q$ is
  a {\em convex} point. The boundary has concave (convex)
  curvature\footnote{We allow infinite curvature.}  at concave
  (convex) points.
\end{definition}

\begin{definition}
  Let $\shapesym$ be a shape and $P\in\Reals^2$. The {\em Euclidean
    distance function} on $\shapesym$, denoted $\distsym_\shapesym$,
  is
  $$
  \distsym_\shapesym (P) = \inf_{Q\in\del\shapesym} \norm{P-Q}_2
  $$

  The {\em boundary support} of $P$, denoted $\bsupp{P}{\shapesym}$, is 
  $$
  \bsupp{P}{\shapesym} = \st{Q\in\del\shapesym:~\norm{P-Q} =
  \distsym_\shapesym (P)}.
  $$
  
  The {\em medial axis} of $\shapesym$, denoted $\MA(\shapesym)$, is
  the set of points $P$ which have two or more closest points on the
  boundary, i.e.  
  $$
  \MA(\shapesym) = \st{P \in \shapesym : \abs{\bsupp{P}{\shapesym}}
  \geq 2}.
$$
where $\abs{\bsupp{P}{\shapesym}}$ denotes the cardinality of the
set.  Note that $\MA(\shapesym)$ does not have to be restricted to the
shape $\shapesym$ and is well-defined on all of $\Reals^2$. Hence,
there is an interior medial axis and an exterior one. Here, we will
only be concerned with the interior one.
\end{definition}

\begin{theorem}\label{theorem:distfunc}
  Let $\shapesym$ be a shape and $P\in\shapesym$. Suppose
  $\abs{\bsupp{P}{\shapesym}} = 1$ and pick the unique
  $Q\in\bsupp{P}{\shapesym}$.  Then,

  \begin{enumerate}
  \item [(a)] $P\not\in\del\shapesym$ implies
    $$
    \grad\distsym_\shapesym (P) = \frac{Q-P}{\norm{Q-P}}
    $$
    where $Q-P$ is the vector from $Q\in\del\shapesym$ to $P$.
  \item [(b)] If $\del\shapesym$ is $C^k$ at $Q$, then
    $\grad\distsym_\shapesym$ is $C^k$ at $P$.
  \end{enumerate}

\end{theorem}
\begin{proof}
  Part (a) is due to \cite[4.8(3)]{Federer59} and (b) is a consequence
  of the more general result by \cite{Krantz81} (see also
  \cite{mather83:dist}).
\end{proof}

\begin{figure}[htbp]
  \centering
  \includegraphics[width=\figwidth]{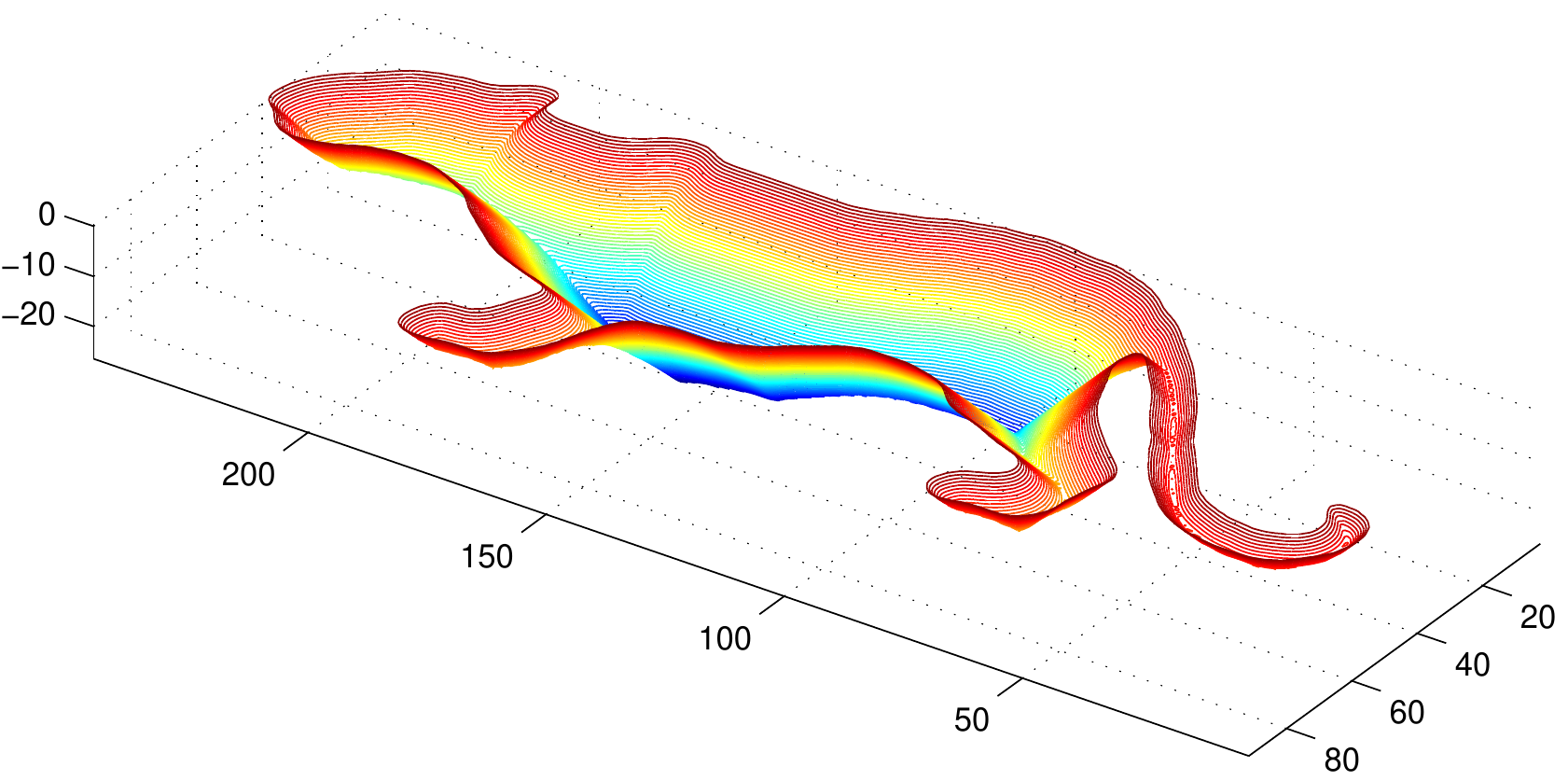}
  \includegraphics[width=\figwidth]{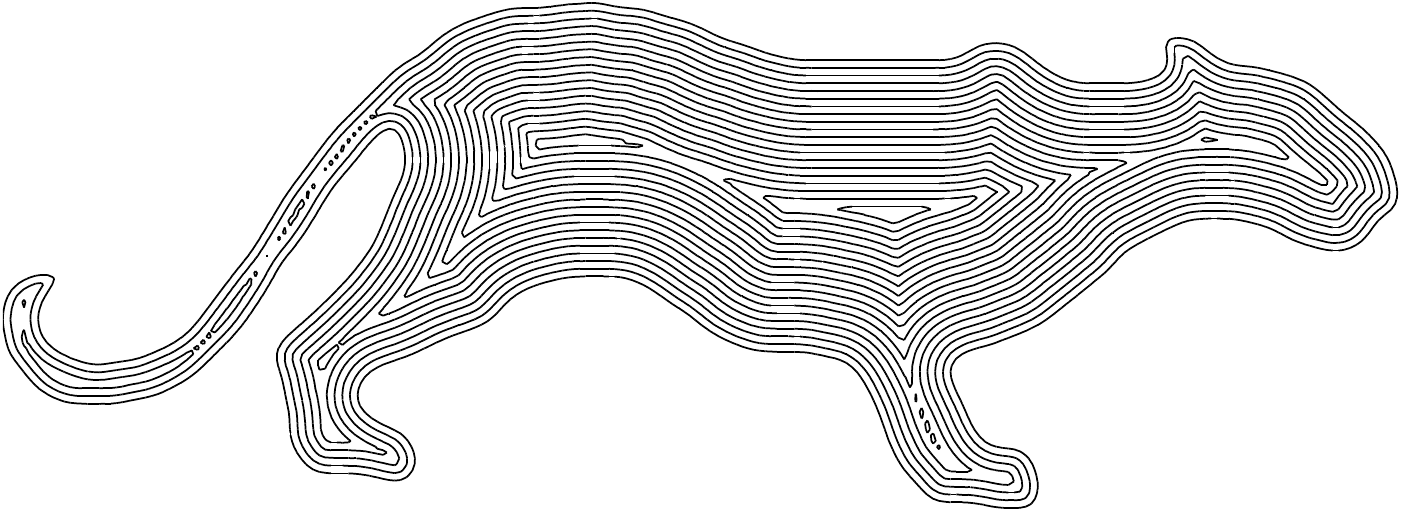}
  \includegraphics[width=\figwidth]{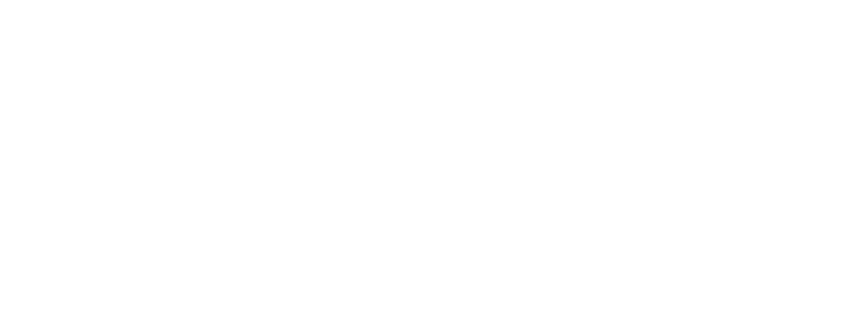}
  \caption[Examples of distance map and medial axis (see
    \cite{ma-hj}).]{Examples of distance map and medial axis (see
    \cite{ma-hj}). {\sc top:} negative distance map
    $-\distsym_\shapesym$, {\sc center:} level sets of $\distsym_\shapesym$,
         {\sc bottom:} medial axis computed as in
         \cite{pdimit:cvpr03}.}
  \label{fig:ma-samples}
\end{figure}

\begin{corol}\label{cor:dist-is-smooth}
  $\distsym_\shapesym (P)$ is smooth at $P\in\shapesym-\MA(\shapesym)$.

\end{corol}
\begin{proof}
  Immediate from \refthm{theorem:distfunc}(b) and the definition of
  shape.
\end{proof}

\begin{theorem}\label{theorem:me-is-thin}
  Let $\shapesym$ be a shape. Then
  \begin{enumerate}
  \item [(i)]   $\MA(\shapesym)$ has no interior, i.e. it is thin.
  \item [(ii)]  $\MA(\shapesym)$ consists of a finite number of connected
    piece-wise smooth curves.
  \item [(iii)] if $P\in\MA(\shapesym)$, $Q\in\bsupp{P}{\shapesym}$
  and $C$ is the center of curvature for $\del\shapesym$ at $Q$, then
  $\norm{P-Q} \leq \norm{C-Q}$ whenever $\norm{C-P}\leq \norm{C-Q}$.
  \end{enumerate}
\end{theorem}
\begin{proof}
  Part (i) is shown in \cite{matheron88:examp_topol_proper_skelet} and
  in \cite{calabi68:skeletons}; (ii) is treated in detail by
  \cite{Choi97}. Part (iii) asserts that if a medial axis
  point is inside the circle of curvature of a point in its boundary
  support, then it cannot be further than the center of curvature.
\end{proof}

\begin{theorem} \label{thm:unique-existence}
  Let $\shapesym$ be a shape. There is a unique $c$ on $\shapesym$
  such that $c=0$ on $\del\shapesym$ and $\grad^2c=-K/D$.
\end{theorem}
\begin{proof}
  See \cite[p.~246]{courant-hilbert62} or Theorem 4.3 of
  \cite{gilbarg83:book} for a more general statement and proof. Also
  see \cite{kannenberg89}.
\end{proof}

\begin{theorem}[Divergence] \label{thm:divergence}
  Let $\del B_\eps(P)$ be a circle or radius $\eps$ centered at
  $P\in\Reals^2$, $\cal{N}$ the inner normals.  Then
  $$
  \grad^2 c(P) = \lim_{\eps\to 0} \int_{\del B_\eps(P)}
  \inner{\grad c}{N} ds .
  $$
\end{theorem}
\begin{proof}
  See p.~151 in \cite{warner:diffman}.
\end{proof}

\begin{definition}
  The $\Theta$-notation for asymptotic behavior of a function is
  defined as:
  
  $$
  \Theta (g(n)) = \st{f(n) : \exists
    c_1,c_2,n_0~\mathrm{positive}~s.t.~\forall n>n_0, 0\leq c_1 g(n)
    \leq f(n) \leq c_2 f(n)} .
  $$
\end{definition}



\subsubsection{Statement of Result}

\refresult{cph-result} is based on the following theorem.

\begin{theorem}\label{theorem:maintheorem}
  Let $\shapesym$ be a shape and $c:\shapesym\to\Reals$ the unique
  function satisfying $c(x,y)=0$ on $(x,y)\in\del\shapesym$ and
  \begin{equation}\label{eq:poisson}
    \grad^2 c = -\frac{K}{D}\quad.
  \end{equation}

  Suppose $P\in\shapesym$ is such that
  $\distsym_{\shapesym}(P) = L = \sup_{\shapesym} \distsym_\shapesym$
  and $Q\in\bsupp{P}{\shapesym}$. Suppose the smallest concave
  curvature radius is $pL$ with $p>0$. Then,
  \begin{enumerate}
  \item [(a)] $c(Q) \in \Theta(L^2)$,
    
  \item [(b)] $\frac{K}{2D} L \leq \abs{\grad c} \leq \frac{K}{D}L
    \frac{2p+1}{p}$,
    
  \item [(c)] $\sup_{\del\shapesym} \abs{\grad c} >
    \sup_{\shapesym-\del\shapesym} \abs{\grad c}$
  \end{enumerate}
\end{theorem}

If an areole is regarded as a discretization of a shape $\shapesym$,
then the discrete approximation behaves as stated in
\refthm{theorem:maintheorem}. Observe that $\grad c$ at $\del\shapesym$ is
perpendicular to the boundary because $c=0$ there; hence, $\grad c(Q)$
points in the direction of $P$ according to \refthm{theorem:distfunc}.

\subsubsection{Organization of the Proof}

Parts (a) and (b) of \refthm{theorem:maintheorem} follow from
\reflem{lemma:bounds}.  The idea of the proof is to find appropriate
bounding functions, one from below $v$ and another $u$ from above,
that sandwich the unique solution $c$ and that take the same values
at the boundary. Thus, $v\leq c\leq u$ everywhere and $\abs{\grad
  v}\leq \abs{\grad c}\leq \abs{\grad u}$ on points where $v=c=u$,
i.e. the boundary. Since the value of $c$ must be the same at the
boundary, its gradient there must be perpendicular to the boundary
which gives the direction as claimed in \refresult{cph-result}(b).
\reflem{lemma:disc} and \reflem{lemma:ubound} give the lower bound and
upper bound constructions and \reflem{lemma:bounds} collects them.

Part (c) of \refthm{theorem:maintheorem} is necessary to show that the
c-vascular strand creation process is well defined.
\refresult{cph-result}(c) is the non-technical version of this claim
which is stated more precisely in \reflem{lemma:decgrad}. The proof is
based on the idea that the boundary may be seen as evolving by
considering level sets of $c$, i.e. points $\gamma_{c_0}$ where
$c(x,y) = c_0$.  The gradient must be perpendicular to this level set
and the solution of \refeq{eq:poisson} inside it follows the same
constraints as the shapes on which the problem is defined. We may move
the level set curve $\gamma_{c_0}$ so that the point on $\gamma_{c_0}$
which is on the gradient curve initiated at the point $Q$ of maximal
gradient on $\gamma_0$ touches $Q$ for small enough $c_0$. Knowing
that the solution on the smaller shape must be strictly smaller than
on the original shape shows that the maximum gradient magnitude must
be strictly decreasing as the curve evolves.  This is true for all
curves, including the evolved ones (i.e.  $\gamma_{c_1}$ for
$c_1>c_0$), so the gradient in the interior of the shape must be lower
than the maximum on the boundary.

\subsubsection{The Proof}


We begin with \reflem{lemma:mondyn} which will be used (indirectly) in
most of the proofs that follow. It states that a discretization of the
dynamic process will always move the concentration values in the same
direction (up or down) if this direction is locally the same for all
discrete points. This fact will be used to prove the next result,
\reflem{lemma:containment}, which states that the equilibrium solution
over a shape completely contained in another shape will be bounded
above by the solution over the larger shape. This holds even if the
initializing function is not smooth.

\begin{lemma}\label{lemma:mondyn}
  Let $c_t = D \grad^2 c + K$ be approximated on a square lattice by
  $p_i$ and its four neighbors $n_j$ by $\approxct=
  \frac{D}{h^2}\paren{\sum_j c(n_j) - 4 c(p_i)} + K$ where $h$ is the
  lattice spacing. Suppose that $\approxct \leq (\geq) 0$ everywhere
  on the domain of definition at time $t_0$. Then

  \begin{enumerate}
  \item[(a)] $c + \tau \approxct$ will also satisfy the inequality if
    $0<\tau \leq \frac{h^2}{4D}$; and

  \item[(b)] the discrete dynamics with such $\tau$ make $c$ decrease
  (increase) monotonically everywhere.
  \end{enumerate}
\end{lemma}
\begin{proof}
  Part b) follows directly from a).
  Let $\alpha = D/h^2$ and $\Lambda_p = \sum_{j=1,4} c(n_j) - 4 c(p)$,
  so $\approxct(t_0, p) = \alpha \Lambda_{p} + K$.
  After the time step $\tau$ the approximation to $c_t$ becomes:
  $$
  \begin{array}{rcl}
    \approxct(t_0+\tau) &=& D \grad^2(c + \tau \approxct(t_0)) + K \spcing\\
    &=& \displaystyle
    \alpha\paren{\sum_j \paren{c(n_j) + \tau \approxct(t_0,n_j) } -
      4 \paren{c(p_i) + \tau \approxct(t_0,p_i)}} + K
    \spcing
    \\
    &=&
    \displaystyle
    \alpha \paren{\Lambda_{p_i} + \tau\paren{ \sum_{j=1}^4
        \approxct(t_0,n_j) - 4\approxct(t_0,p_i)}} + K
    \spcing\\
    &=&
    \displaystyle
    \alpha \paren{\Lambda_{p_i} + \tau\paren{ \sum_{j=1}^4
        (\alpha\Lambda_{n_j}+K) - 4(\alpha\Lambda_{p_i}+K)}} + K
    \spcing\\
    &=&
    \displaystyle
    \alpha \paren{\Lambda_{p_i} + \tau\alpha\paren{ \sum_{j=1}^4
        \Lambda_{n_j} - 4\Lambda_{p_i}}} + K
    \spcing\\
    &=&
    \displaystyle
    \alpha\Lambda_{p_i} (1-4\tau\alpha) +\alpha\tau\alpha \sum_{j=1}^4
    \Lambda_{n_j} + K
    \spcing\\
  \end{array}
  $$
  
  \noindent
  Now, all of the $\alpha\Lambda_{n_j}+K\leq 0$ and
  $\alpha\Lambda_{p_i}+K\leq 0$ by assumption. Also,
  $0<4\tau\alpha\leq 1$. So choosing the largest $\Lambda_{n_j}$ and
  replacing for the other three bounds the value above (since
  $\Lambda_{n_j}<0$), shows that $\tau$ is used in a linear
  interpolation between two non-positive numbers. This finishes the
  claim.
\end{proof}

\begin{lemma}\label{lemma:containment}
  Let $u_t=D\grad^2u+K=0$ over $\shapesym$ with $u=0$ on
  $\del\shapesym$. If $\shapesym'\subset \shapesym$ and $D\grad^2
  c+K=0$ on $\shapesym'$ with $c=0$ on $\del\shapesym'$ then $c\leq u$
  on $\shapesym'$. If $\shapesym\not\subset\shapesym'$, then $c < u$
  everywhere on $\shapesym'-\del\shapesym'$.
\end{lemma}
\begin{proof}
    If $u=0$ over $\shapesym$, then $u_t = K > 0$ and, by
    \reflem{lemma:mondyn}b, $u>0$ in the interior
    $\shapesym-\del\shapesym$ after any non-zero time step. Thus, at
    equilibrium, $u$ will satisfy the dynamics everywhere on
    $\shapesym'$ except possibly on $\del\shapesym'$ where it may need
    to be lower because it the boundary conditions. If a boundary
    point is lowered in the discretization of the problem, any
    neighbor will see its $\grad^2 u$ decrease and strictly decrease
    if the neighbor is not moved. This holds at any time step and will
    affect all points after sufficiently long time because they are
    all connected. If $\shapesym\not\subset\shapesym'$, then $u$ is
    not a solution ($u>0$ somewhere on $\del\shapesym'$) and the
    dynamics on $\shapesym'$ will strictly monotonically lower it
    everywhere.
  \end{proof}

Now we show what the solution looks like in one dimension,
\reflem{lemma:1D}, and then we turn to two special shapes: the circle
(\reflem{lemma:disc}) and the open doughnut
(\reflem{lemma:odoughnut}). These shapes will be instrumental in
providing the lower and upper bounds needed later on.

\begin{lemma}\label{lemma:1D}
  Suppose the domain is the segment $[0,L]$, $\c1D(0)=0$ and
  $\c1D'(L)=0$. Then the solution to \refeq{eq:poisson} is
  $$\c1D(r) = \frac{K}{D} \paren{ -\frac{r^2}{2} + r L } $$
\end{lemma}
\begin{proof}
  By inspection since the solution is unique: $\del^2/(\del
  r)^2[\c1D(r)] = -K/D$.
\end{proof}

\begin{lemma}[Disc] \label{lemma:disc}
  Let the shape be a circle of radius $L$ centered at the
  origin. Suppose $c(x,y) = 0$ on the boundary where
  $x^2+y^2=L^2$. Then the solution to \refeq{eq:poisson} is
$$c(x,y) = \frac{1}{2} \c1D\paren{L-\sqrt{x^2+y^2}}$$
  and 
  $$\abs{\grad c(x,y)} = -\frac{1}{2}\frac{K}{D} \sqrt{x^2+y^2}$$
\end{lemma}
\begin{proof}
  By inspection since the solution is unique. We see that $c_{xx} =
  c_{yy} = -\frac{1}{2}K/D$, so \refeq{eq:poisson} is satisfied.
\end{proof}

\begin{figure}[tbp]
  \centering
  \includegraphics[width=\figwidth]{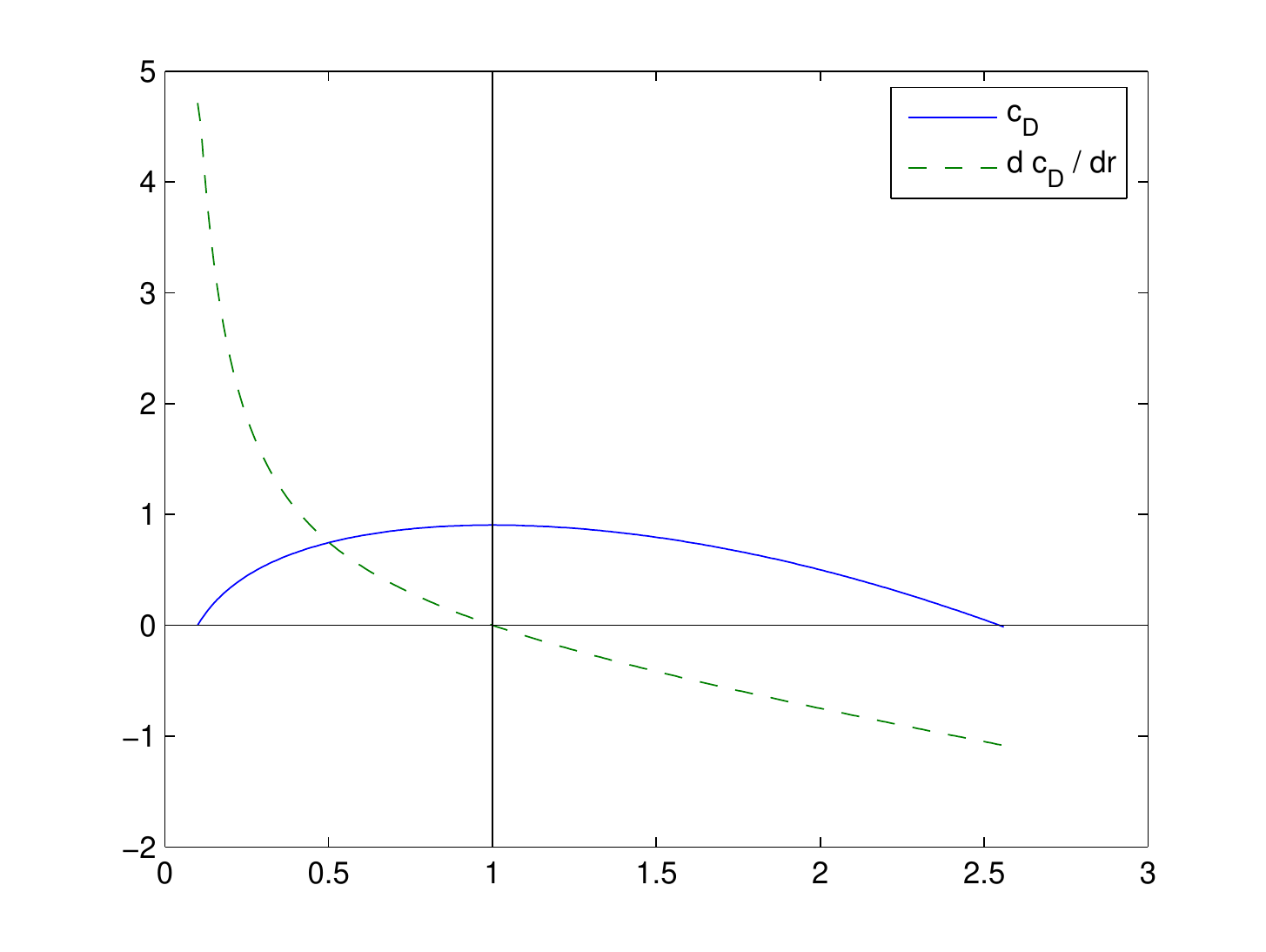}
  \caption{A plot of $c_D(r; 0.1, 0.9)$ and $\frac{d c_D}{d
      r}$. Notice that $c_D$ is increasing from $l$ to $l+L$.}
  \label{fig:cD}
\end{figure}

\begin{lemma}[Open Doughnut] \label{lemma:odoughnut}
  Let $0 \leq l < l+L$ be the radii of two circles centered at the
  origin. Suppose $c(x,y) = 0$ on the boundary $x^2+y^2=l^2$ and
  $\grad c(x,y)=0$ for $x^2+y^2=(l+L)^2$.
  Then the solution to \refeq{eq:poisson} is
  $$c(x,y) = c_D\paren{\sqrt{x^2+y^2}}$$
  where
  \begin{equation}
    \label{eq:c_D}
    c_D(r) = c_D (r ; l,L)= \frac{K}{D}\paren{\frac{1}{4}\paren{l^2-r^2} +
      \frac{1}{2}{\ln\paren{\frac{r}{l}}}(l+L)^2}
  \end{equation}
  Further, $c(x,y)\geq 0$ for $l^2\leq x^2+y^2 \leq (l+L)^2$.
\end{lemma}
\begin{proof}
  By inspection since the solution is unique. Since $c=0$ at the inner
  boundary and $\grad c$ points radially toward the outer boundary,
  the values of $c$ are increasing radially in $l^2\leq x^2+y^2 \leq
  (l+L)^2$.
\end{proof}

The next two results (\reflem{lemma:error} and
\reflem{lemma:convc-ondisc}) are technical assertions used in the
proof of the Upper Bound Lemma (\reflem{lemma:ubound}). This is the
last result needed to prove \reflem{lemma:bounds} and, therefore,
parts (a) and (b) of \refthm{theorem:maintheorem}.
\begin{lemma}\label{lemma:error}
  Let $\shapesym$ be a shape $P\in\shapesym$, and
  $Q\in\bsupp{P}{\shapesym}$. Let $\Sigma$ be the circle of curvature
  of $\shapesym$ at $Q$. Then $\distsym_\shapesym(P') =
  \distsym_{\Sigma}(P') + O(\eps^3)$ for $\norm{P-P'}=\eps$, and over
  a circular region $R$ of radius $\eps$ centered at $P$
  $$
  \lim_{\eps\to 0}
  \frac{1}{\pi\eps^2}\int_{\del R} \inner{\frac{\del}{\del r}c_D(\distsym_\shapesym)
  \grad \distsym_\shapesym  }{\mathcal{N}} ds = 
  \lim_{\eps\to 0} 
  \frac{1}{\pi\eps^2}\int_{\del R} \inner{\frac{\del}{\del r}c_D(\distsym_\Sigma)
  \grad \distsym_\Sigma  }{\mathcal{N}} ds
  $$
\end{lemma}
  \begin{proof}
    Write the second order approximation
    $\distsym_\shapesym=\distsym_\Sigma+O(\eps^3)$ and $\grad
    \distsym_\shapesym = \grad \distsym_\Sigma + O(\eps^2)$. In a
    circular neighborhood $R$, the limit becomes:
    $$
    \lim_{\eps\to 0} \frac{1}{\pi\eps^2} \int_{0}^{2\pi\eps}
    \inner{\frac{\del}{\del r}c_D(\distsym_\shapesym) \grad \distsym_\shapesym }{\mathcal{N}} ds =
    \lim_{\eps\to 0}  \int_{0}^{2\pi} \frac{1}{\pi\eps}
    \inner{\frac{\del}{\del r}c_D(\distsym_\shapesym) \grad \distsym_\shapesym }{\mathcal{N}} ds
    $$
    
    \noindent
    Now, $\frac{\del}{\del r}c_D(\distsym_\Sigma + O(\eps^3)) = \frac{\del}{\del
      r}c_D(\distsym_\Sigma) + O(\eps^3)$ by inspection of $\frac{\del}{\del
      r}c_D(r)$, which gives:
    
    $$
    \begin{array}{rcl}
      \frac{1}{\pi\eps} {\frac{\del}{\del r}c_D(\distsym_\shapesym) \grad \distsym_\shapesym
      } &=&
      \frac{1}{\pi\eps}
      {\frac{\del}{\del r}c_D(\distsym_\Sigma+O(\eps^3)) \paren{\grad
          \distsym_\Sigma+O(\eps^2) }  }  
      \spcing \\

      &=&
      \frac{1}{\pi\eps}
      \paren{\frac{\del}{\del r}c_D(\distsym_\Sigma)\grad \distsym_\Sigma + \frac{\del}{\del r}
      c_D(\distsym_\Sigma)O(\eps^2) + O(\eps^3)\grad \distsym_\Sigma + O(\eps^5)} 
    \spcing \\ 
    &\to& \displaystyle
    \lim_{\eps\to 0} \frac{1}{\pi\eps}\frac{\del}{\del r}c_D(\distsym_\Sigma)\grad \distsym_\Sigma
    \end{array}
    $$
  \end{proof}

\begin{lemma}\label{lemma:convc-ondisc}
  Suppose the shape $\shapesym$ is the disc as in \reflem{lemma:disc}
  with radius $l+L$ and $c$ is the solution. Then,
  $$
  u(x,y) = c_D(l+L-\sqrt{x^2+y^2})
  $$
  satisfies $\grad^2 u < \grad^2 c = -K/D$ for all points except
  the center.
\end{lemma}
  \begin{proof}
    Write $f = u - c$ where $c$ is the solution for the disc from
    \reflem{lemma:disc}. Letting $R(x,y) = \sqrt{x^2+y^2}$
    $$
    \begin{array}{rcl}
    f(R(x,y)) &=& c_D(l+L-R) - \frac{1}{2}\c1D(l+L-R) 
    \spcing\\
    &=& 
    \frac{K}{D} \paren{\frac{l^2}{4}
    +\frac{1}{2}\paren{\ln\paren{\frac{l+L-R}{l}} 
    (l+L)^2 - (l+L-R) (l+L)} }
    \end{array}
    $$
    
    \noindent
    A direct calculation shows that 
    $$
    \grad^2 f(x,y) = -\frac{1}{2}\frac{K}{D}
    \frac{l+L}{\sqrt{x^2+y^2}}
    $$

    \noindent
    which demonstrates that $f(r) < 0$ everywhere and
    $\lim_{r\to0}f(r)=-\infty$ in the center of the disc.  Therefore,
    $\grad^2~u~=~\grad^2~f~+~\grad^2 c < -K/D$ because $\grad^2 c =
    -K/D$.
  \end{proof}




\begin{lemma}[Upper Bound] \label{lemma:ubound}
  Let the conditions of \refthm{theorem:maintheorem} hold. Define
  $u(x,y) = 2c_D(l+\distsym_\shapesym(x,y))$ with $l=p L$. Then $u\geq
  c$.
\end{lemma}
  \begin{proof}
    We shall show that any discretization $\tilde u$ with spacing
    $h<h_0$ (for some $h_0>0$) of $u$ will satisfy $D\grad^2 \tilde{u} + K
    \leq 0$. Thus, \reflem{lemma:mondyn} shows that a dynamical
    process initialized with $u$ will decrease $u$ everywhere with
    each time step and, by \refthm{thm:unique-existence}, it should
    converge to $c$. Hence, $c\leq u$.
    
    First, we treat non-medial axis points.  Let $P=(x,y)\in
    \shapesym$ (not on the medial axis) and $Q\in\del\shapesym$ which
    is closest to $P$, i.e. $\norm{P-Q} = \distsym(P)$. Suppose
    $\del\shapesym$ near $Q$ is approximated by the circle of
    curvature at $Q$. Thus, \reflem{lemma:error} applies and, by the
    Divergence \refthm{thm:divergence}, $\grad^2 u(P)$ is the same as
    if the boundary were a circle at $Q$. Following
    \refthm{theorem:me-is-thin} MA2, there are three cases: (a) $P$
    outside the circle, (b) $P$ inside, and (c) the circle has
    infinite radius -- it is a line segment.  \reflem{lemma:odoughnut}
    shows that $\grad^2 u(P) = -2K/D <-K/D$ which takes care of (a),
    and (b) is covered by \reflem{lemma:convc-ondisc}. If the boundary
    is locally a straight line, then 
    \begin{equation}
      \label{eq:lambda_cD}
    \grad^2 u = \frac{\del^2}{\del r^2} 2c_D(l + r) = -\frac{K}{D}
    \paren{ 1+ (l+L)^2/r^2}
    \end{equation}
    with $l < r < l+L$ and $r$ is in the
    direction of the gradient. So, $\grad^2 u(P) < -K/D$ for all
    $P\in\shapesym-\MA(\shapesym)$.
    
    Now suppose that the of $u$ on $\shapesym$ are sampled on a
    discrete square lattice with spacing $h>0$. There, $\grad^2 u(p)$
    is approximated by the formula for $\Lambda_p$ in the proof of
    \reflem{lemma:mondyn}. The error of the approximation is $O(h^2)$
    (see \cite{as:mathfuncs-handbook}). Thus,
    from the above, $h$ may be chosen so that $\Lambda^{u}_P < -K/D$
    for $P\in\shapesym$ further than $h$ from $\MA(\shapesym)$.
    
    Let $\distsym_{\Sigma_P}$ to be the distance function from the
    circle of curvature at the boundary point corresponding to
    $P\not\in \MA(\shapesym)$. Hence, if $\norm{P'-P}=h$, then
    $\distsym_{\Sigma_P}(P') = \distsym_\shapesym(P') + \eps_{P'}$
    where $\eps_{P'} = O(h^3)$. Set 
    $$
    \eps_P = \sup_{\norm{P'-P}=h}\abs{\eps_{P'}} \quad
    \mathrm{and}  \quad
    \eps = \sup_{P\in\shapesym-\MA(\shapesym)} \eps_{P'}
    $$
    where $0\leq\eps < h$ for small enough $h$. So, choose such an $h$
    and define
    $$
    u_h (x,y) = 2 c_D (l + \distsym_\shapesym(x,y) ; l, L + 2 h)
    $$
    and notice that we may refine the grid (i.e. choose $h$
    smaller) and the above properties will still hold. Thus, refine
    $h$ if necessary to make $\Lambda^{u_h}_P < -K/D $ on shape points
    further than $h$ from the medial axis. Refine it further to
    $\Lambda^{u_h}_P < -K/D $ on an open doughnut with $l$ inner radius
    and $L+2 h$ outer radius. Make sure that $h$ is small enough so
    that $c_D(r+h) - 2c_D(r) + c_D(r-h) < -K/D$ (this is needed in the
    tangent line construction below), which is possible because of
    \refeq{eq:lambda_cD}.
    
    Now we show that this also makes $\Lambda^{u_h}_P < -K/D $ for
    points on the medial axis and those closer than $h$ from it. Pick
    such a $P$ and let consider $Q\in\bsupp{P}{\shapesym}$. If $Q$ is
    concave, then approximate the boundary by its circle of curvature
    and look at $u_h^\Sigma = 2 c_D (l + \distsym_\Sigma(x,y) ; l, L +
    2 h)$ . A neighbor $N$ of $P$ used in $\Lambda^{u_h}_P$ satisfies
    $\distsym_\Sigma(N) \geq \distsym_\shapesym(N)$ (because $Q$ is
    concave and the difference is no more than $\eps$. Hence,
    $u_h^\Sigma (N) \geq u_h(N)$ since $c_D$ is increasing until
    $l+L+2h$. Further, $\distsym_\Sigma(P) = \distsym_\shapesym(P)$
    because the circle of curvature touches $\del\shapesym$ at
    $Q$. Hence, $\Lambda^{u_h}_P < \Lambda^{u^\Sigma_h}_P < -K/D $
    because $P$ is not a medial axis point for $\Sigma$.
    
    If, on the other hand, there is a concave
    $Q\in\bsupp{P}{\shapesym}$, then instead of the circle of
    curvature we may take the tangent line $\ell_P$ at $Q$ define
    $\distsym_{\ell_P}$ exactly similarly to $\distsym_{\Sigma_P}$
    above. Refine $h$ so that any point $P'$ for which $\norm{P-P'}=h$
    is closest to a point $Q'\in\ell_P$ that lies outside the
    $\shapesym$ or on $\del\shapesym$.\footnote{This must be possible
      since $Q$ is convex.}. Thus, as before,
    $\distsym_{\ell_P}(N)\geq \distsym_\shapesym(N)$ for any $N$ near
    $P$, i.e. $\norm{N-P}=h$. Therefore, $\Lambda^{u_h}_P <
    \Lambda^{u^\Sigma_h}_P < -K/D $.

    Finally, \reflem{lemma:containment} shows that $u_h > c$ from
    which we conclude that $u \geq c$ since $\lim_{h\to 0} u_h \to u$.
  \end{proof}

\begin{remark}
  The function $u(x,y)$ need not be smooth on $\shapesym$. In fact, it
  will fail to have fist derivatives on the medial axis of most shapes.
\end{remark}

\begin{lemma}\label{lemma:bounds}
  Let the conditions of \refthm{theorem:maintheorem} hold.  Then $c =
  \Theta(L^2)$. Further, if $P$ is the center of the largest inscribed
  circle and $Q$ a point on the boundary of the shape and the circle,
  then $\grad c (Q)$ points toward $P$ and
  $$
  \frac{K}{2D}L\leq\abs{\grad c}\leq \frac{K}{D} L \frac{2 p +
    1}{p} .
  $$
\end{lemma}
\begin{proof}
  The Disc \reflem{lemma:disc} gives the lower bound function and the
  Upper Bound \reflem{lemma:ubound} the rest. The gradient points in
  the direction of the normal to the boundary because $c=0$ on
  $\del\shapesym$. The magnitude follows from a simple calculation of
  $\del/\del r [c_D(r)]$ at $r=pL$ (see the Open Doughnut Lemma for
  the definition of $c_D(r)$).
\end{proof}



Finally, the following result completes the proof of
\refthm{theorem:maintheorem}.

\begin{lemma}[Decreasing Gradient] \label{lemma:decgrad}
  Let $c$ satisfy the Poisson equation (\refeq{eq:poisson}) on
  $\shapesym$ and $c=0$ on $\del\shapesym$. Then 
  $$
  M=\sup_{\shapesym}\abs{\grad c}=\sup_{\del\shapesym}\abs{\grad c}
  $$
  and
  $$
  \abs{\grad c(x,y)} < M, \quad (x,y)\in\shapesym-\del\shapesym~.
  $$
\end{lemma}
  \begin{proof}
    Let $\gamma_{c_0} = \st{(x,y)\in\shapesym : c(x,y)=c_0}$. A number
    $0 < c_0 < \sup_\shapesym c$ must exist since $c>0$ on $\shapesym
    - \del\shapesym$ by \reflem{lemma:mondyn}b. Let the shape
    $\shapesym'$ be defined by $(x,y)\in\shapesym$ such that
    $c(x,y)\geq c_0$. Hence, $\shapesym' \subset \shapesym$ and
    $\shapesym \not\subset \shapesym'$. The boundary
    $\del\shapesym'=\gamma_{c_0}$ is regular, so there is a unique $v$
    satisfying \refeq{eq:poisson} on $\shapesym'$ with $v=0$ on
    $\del\shapesym'$. Thus, $v = c-c_0$ and $\grad v = \grad c$ on
    $\shapesym'$.
    
    Thus, $\gamma_0 = \del\shapesym$ and $\gamma_{c_0}$ is connected
    for small enough $c_0$ (because $\shapesym$ is the closure of an
    open set). Further, if $c_0<\eps_0$ for some $\eps_0>0$, then
    $\gamma_{c_0}$ is a smooth curve because $\grad_{\gamma_{c_0}}
    c=0$ on $\gamma_{c_0}$, $c$ is at least twice differentiable, and
    $\grad c\neq 0$ when taken over $\shapesym$ on points of
    $\gamma_0$. In fact, $\grad c$ is perpendicular to the curve
    $\gamma_0=\del\shapesym$. Let $Q\in\gamma_0$ be such that $\grad
    c(Q) = \sup_{\del\shapesym} \abs{\grad c}$. Let $\beta$ be the
    integral curve segment starting at $\beta(0)=Q$ with tangents in
    the direction of $\grad c$ and such that
    $\beta(1)\in\gamma_{c_0}$.  Since $\grad c$ is perpendicular to
    $\gamma_0$, $\beta(1)$ will be the closest point to $\gamma_0$
    from $\gamma_{c_0}$ for small enough $c_0$. Thus, $\gamma_{c_0}$
    may be translated so that $\beta(1)$ touches $Q$ ensuring that
    $\gamma_{c_0}$ is completely contained in $\shapesym$; denote this
    translated curve by $\gamma_{c_0}'$.
        
    The solution $v'$ to \refeq{eq:poisson} on $\gamma_{c_0}'$ and its
    interior must be the translated $v$. By \reflem{lemma:containment}
    $v' < c$ everywhere except on $\gamma_0\cap\gamma_{c_0}'$ (e.g. at
    $Q$) where $v'=c$. Hence, $\abs{\grad v'(Q)} < \abs{\grad u(Q)}$.
    Therefore, $\abs{\grad u (Q)}$ is strictly decreasing in the
    direction of $\grad u(Q)$.
  
    \begin{figure*}[htbp]
      \centering
      \begin{tabular}{ccc}
        \includegraphics[width=0.3\textwidth]{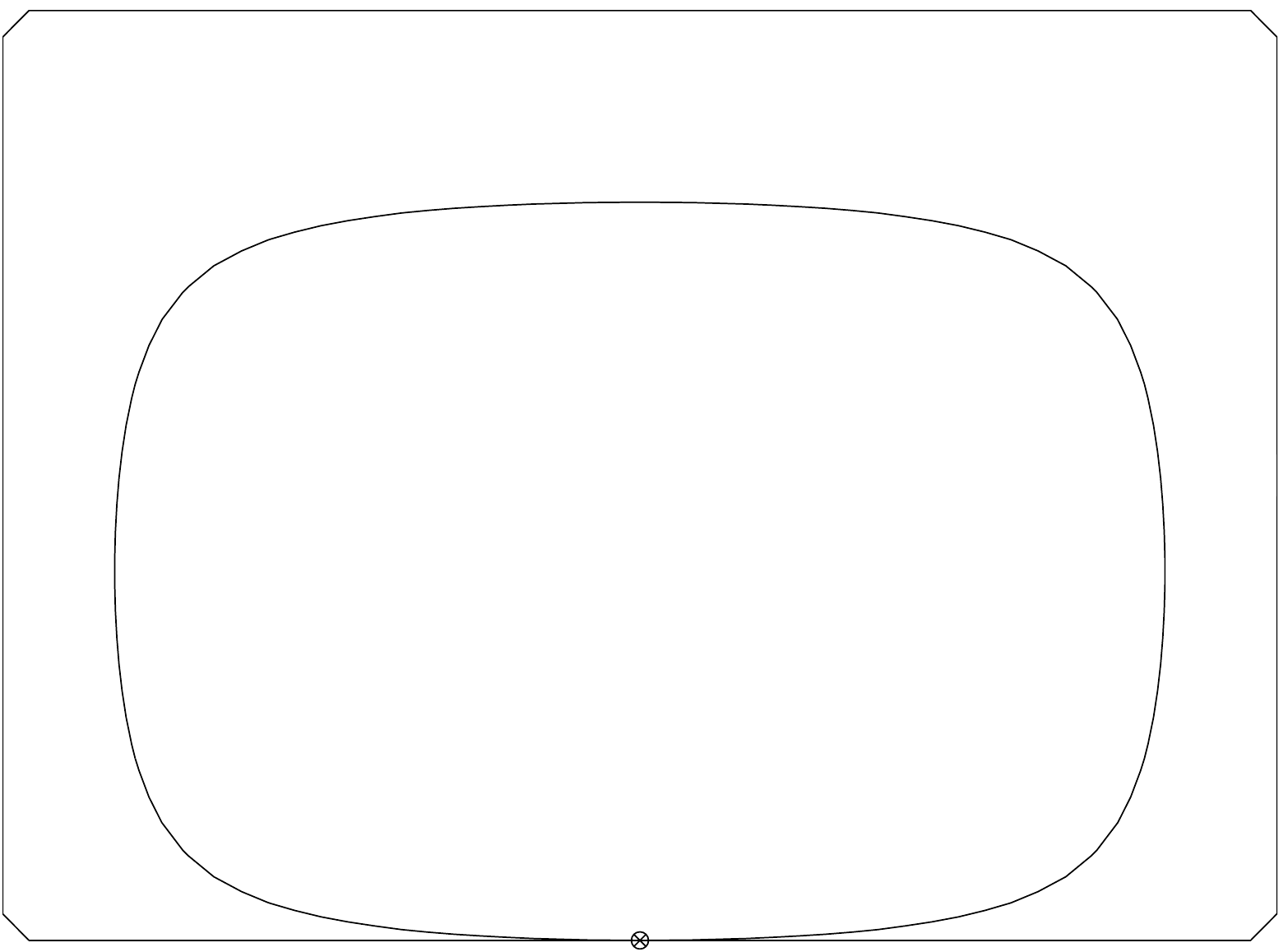}
        &
        \includegraphics[width=0.3\textwidth]{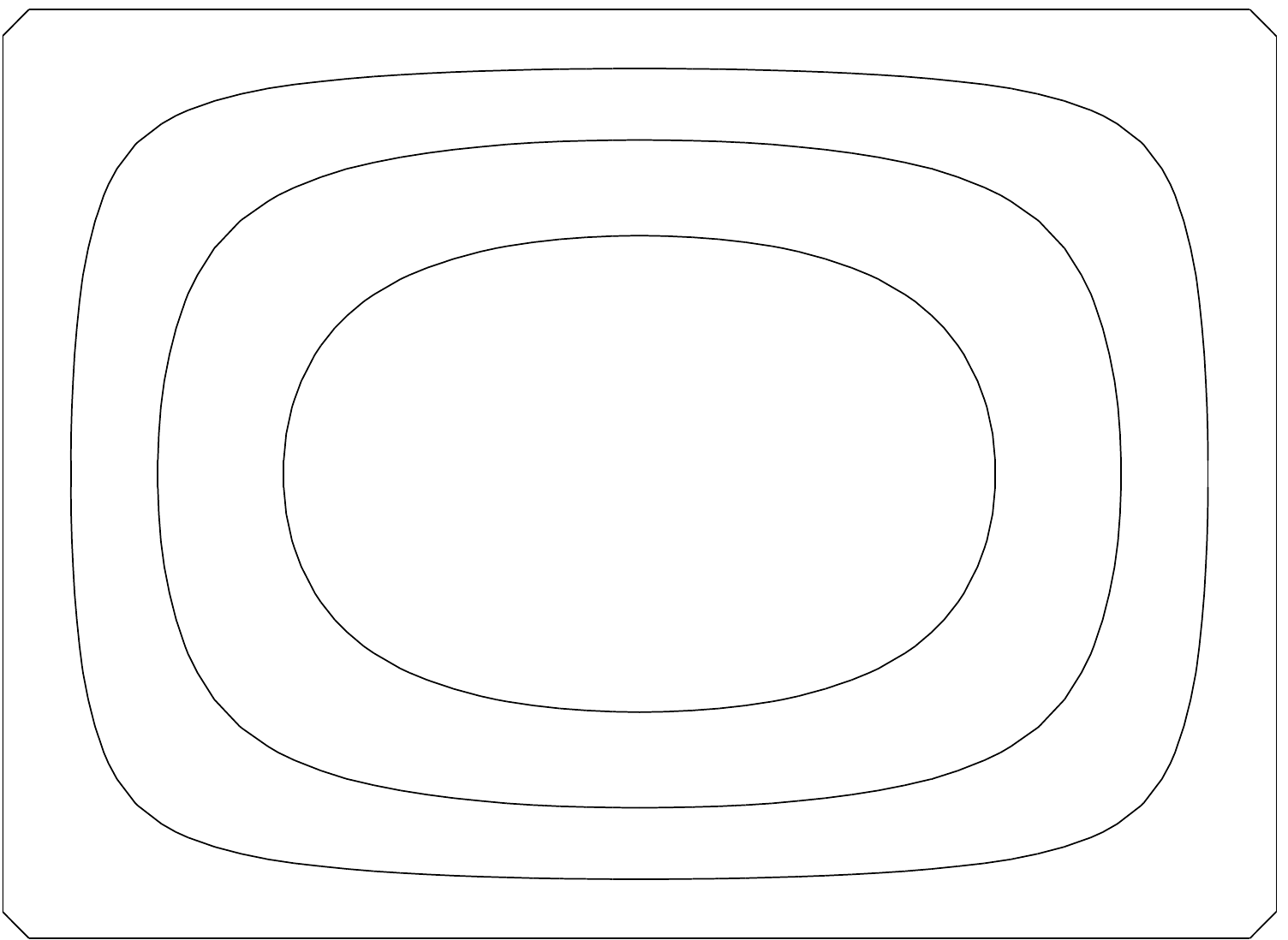}
        &
        \includegraphics[width=0.3\textwidth]{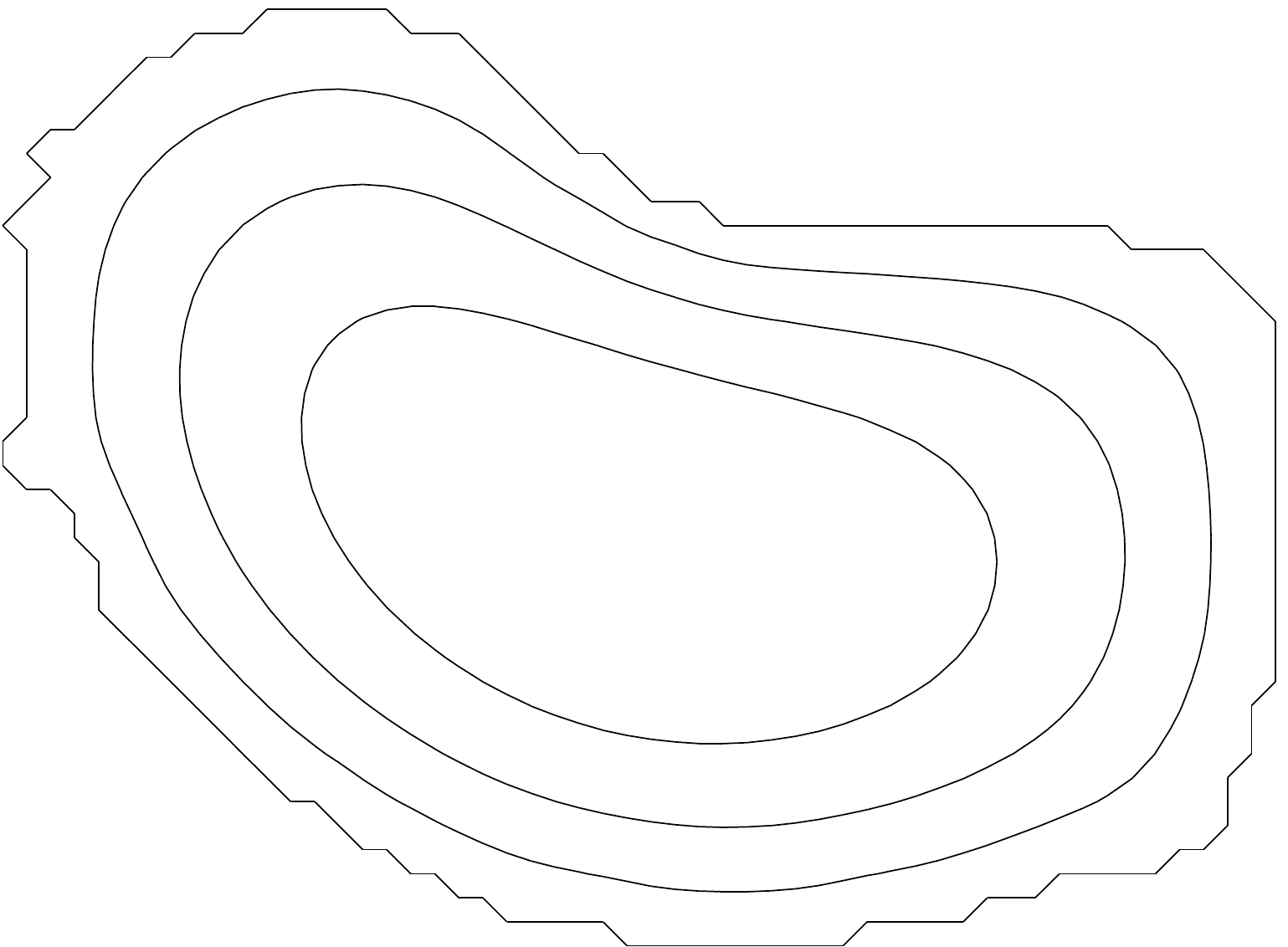}
        \\ 
        (a) & (b) & (c)
    \end{tabular}
    
      \caption{Level sets of $c$. \FP{a,b} Rectangular shape. \FP{c}
        Areole from \cite{scarpella:genes-for-procambium}.}
      \label{fig:isocurves}
    \end{figure*}

  \end{proof}

\section{Proportional Destruction Hypothesis: Helmholtz Model}

\subsection{Background}

In \cite{pdswz:paper} we noted that new veins may emerge
simultaneously in areoles of drastically different sizes. We argued
that the most parsimonious explanation of this phenomenon is that
auxin is produced at a rate that is constant for each cell, regardless
of that cell's size. But we did not develop the question of auxin
destruction beyond assuming that veins drain the hormone in such a way
that it is effectively depleted from the areole. A more natural
assumption emerges from considering recent molecular work.

Auxin is involved in a variety of processes that take place in
plant cells at all times, so a portion of the free IAA is constantly
being depleted. For example, it has recently been shown
\cite{kepinski-leyser-2005,dharmasiri-etal-2005} that the early
response genes are activated by auxin through an increased degradation
of promoter inhibitors. Thus, auxin binds to a TIR1 protein which is
instrumental in tagging inhibitor proteins (Aux/IAA) for degradation:
the larger the concentration of auxin, the more effective the
degradation.
In addition, the hormone ``is readily
conjugated to a wide variety of larger molecules, rendering it
inactive. Indeed, the majority of IAA in the plant is in the form of
inactive conjugates. Auxin conjugation and catabolism can therefore
decrease active auxin levels.''  \cite[p.~853]{teale-etal-2006}. For
these reasons, we propose the Proportional Destruction Hypothesis:
\begin{hypothesis}[Proportional Destruction]
  \label{hypothesis:pdh}
  Free auxin levels are constantly being depleted at a rate
  proportional to the available hormone levels.
\end{hypothesis}

\noindent
The constant of proportionality will be denoted by $\alpha$ and
referred to as the {\em destruction constant}. 

\subsection{Helmholtz Model Definition}


\begin{table}[t]
  \centering

\begin{tabular}{cp{1ex}c}
  \begin{tabular}[t]{|p{0.9\textwidth}|}
    \hline
    {\bf Helmholtz Model}\\
    \hline
    {\bf Definitions}\\
    
    {\em Ground Cell:}  Diffusion coefficient $D_{g}$; production
    $\rho = K / S$. \\
    
    {\em C-Vascular Cell:} At least one interface has diff. coef. $D_{v} >
    D_{g}$;  production $\rho = K / S$.\\
    
    \hline
    {\bf Cell Functions} (Program) \\

    {\sc CF1:} Produce substance $s$ at constant rate $K$ and destroy
    it at a rate $\alpha c$. \\
    
    {\sc CF2:} Measure $c$ and $\Delta c$ through interfaces.\\

    {\sc CF3:} Diffuse $s$ through interfaces.\\
    
    {\sc CF4:} When $\Delta c > \tau$ through interface $I$, change its
    diffusion coefficient to $D_v$. \\
    
    \hline
  \end{tabular}
  
\end{tabular}
  \caption[Helmholtz model definitions.]{Helmholtz model.  Terms
    definitions: $K$ is the per-cell substance production
    rate (mass/time); $S$ is the size of the cell (volume); $\rho$ is
    the per-volume production rate; $D_x$ are diffusion coefficients
    related to the permeability of cell interfaces; $c$ is the
    concentration of the substance inside the cell, $\alpha$ is the
    destruction constant, and $\Delta c(I)$ is the difference of
    concentration through the interface $I$; $\tau$ is a threshold.}
  \label{tab:helmholtz-model}
\end{table}

The Proportional Destruction Hypothesis leads to the updated
formulation of our model shown in \reftab{tab:helmholtz-model}.  We
shall assume that the only mode of auxin transport is diffusion
according to Fick's Law. Our constant production of $K$
(kg$\cdot$s$^{-1}$) is replaced by a speed of increase in
concentration equal to $K / S(i)$ where $S(i)$ is the size of cell
$i$. The proportional destruction of auxin implies a decrease in
concentration given by $-\alpha m(i)/S(i) = -\alpha c(i)$ with $c(i)$
denoting concentration. Therefore, the manner in which the auxin
concentration changes with time within each cell can be written as:
\begin{equation}
  \label{eq:production-destruction}
  \frac{\del c}{\del t} = D \grad^2 c + \frac{K}{S} - \alpha c ~.
\end{equation}

The equilibrium of this dynamical system, when $c_t=0$, is therefore
described by an inhomogeneous Helmholtz equation, which is why we refer
to this formulation as the Helmholtz Model. 

\subsection{Analysis of the Helmholtz Model}
\label{sec:zero-flux}

It turns out that this relatively small modification to the Poisson
model greatly improves the model's descriptive power. We now establish
the mathematical results that make this claim concrete before we turn
to the discussion of biological experiments in the next section. We
demonstrate that the same type of distance information is captured by
the Helmholtz model in \refprop{prop:1} and
\refprop{prop:asin-cstbdry}, but that it can be obtained in more than
one fashion if the destruction constant is small. If this constant is
large, we show that a different kind of distance becomes available to
the discrete system -- one given by the logarithm of concentration.






\begin{prop} \label{prop:1}
Consider the dynamical system
\begin{equation}
  \label{eq:alphadisc} \label{eq:1}
  \partderiv{c}{t} = D \grad^2 c + \rho_\shapesym - \alpha c .
\end{equation}

Suppose that it acts over a domain $\shapesym$ which a shape as in
\refthm{theorem:maintheorem} and on which we impose a zero-flux
boundary condition (Neumann). Let $\rho_\shapesym: \shapesym \to
\Reals$. Then the following holds.

\begin{enumerate}
\item[(a)] If $\alpha>0$, then $\lim_{t\to \infty} c = c_\alpha$ for a
  unique steady-state $c_\alpha$.
  
\item[(b)] Let $\alpha=0$ and $R = \int \rho_\shapesym \d{\shapesym} /
  \int_{}\d{\shapesym}$ be the average production. Then
  $\lim_{t\to\infty} c_t = R$ and $c$ converges to $c_\alpha +
  \mathrm{cst.}$ whenever $R=0$.  Further, $\grad c_\alpha$ is unique
  even when $R\neq 0$.
  
\item[(c)] If $A,B\in \Reals$, then the transformation $\rho_\shapesym
  \mapsto A \rho_\shapesym + \alpha B$ induces a unique transformation
  of the steady state $c_\alpha \mapsto A c_\alpha + B$ and vice versa.
  It follows that the gradient of $c_\alpha$ is only affected if
  $A\neq 1$: $\grad c_\alpha \mapsto A \grad c_\alpha$.
\end{enumerate}
  
\begin{remark}
  In part (c), if the destruction term is not linear, e.g. $\alpha c +
  \beta c^2$, then the gradient might be affected by $B$ as well.
\end{remark}

\end{prop}

\begin{proof}
  Parts (a) and (b). To show existence we prove that the dynamical
  system achieves $c_t=0$. Consider the dynamical system $c_{tt} =
  D\grad^2 c_t - \alpha c_t$. The boundary conditions are inherited:
  since no flux goes through the boundary, there must be no change of
  concentration in time, i.e.  $\grad c_t\cdot \v{n}=0$ on $\del
  \shapesym$. The unique solution of this system is $c_t=0$.
  
  To prove uniqueness, suppose $u_1$ and $u_2$ both satisfy the
  equation given the boundary conditions and $c_t=0$.  Thus
  $$
  D\grad^2 u_1 + \rho_\shapesym -\alpha u_1 =   D\grad^2 u_2 +
  \rho_\shapesym -\alpha u_2 
  $$
  which gives rise to $D\grad^2 v - \alpha v = 0$ where $v=u_1-u_2$
  and $\grad v\cdot \v{n}=0$ where $\v{n}$ is the normal to the
  boundary.  Since $v$ is elliptic and $\alpha>0$, $v$ vanishes
  everywhere and uniqueness follows (see \cite[p.~329 and
  321]{courant-hilbert62}). The same reference shows that if
  $\alpha=0$, then this uniqueness is up to an additive constant
  $u=u_1+\mathrm{cst}$; that is, only $\grad u$ is unique.
  
  Now to show the convergence in (b) whenever $R=0$, note that $c_{tt}
  = D\grad^2 c_t$ assuming $\alpha=0$. This has a steady-state s.t.
  $c_t=\mathrm{cst.}$ everywhere. Also, $\int c_t = \int
  \rho_\shapesym\d\shapesym$ which shows that $c_t=R$.

  Part (c). Let $c_\alpha$ satisfy \refeq{eq:alphadisc} for $c_t=0$
  and a production function $\rho^{(\alpha)}_\shapesym$. Then, $D
  \grad^2 c_\alpha - \alpha c_\alpha = - \rho^{(\alpha)}_\shapesym$.
  Suppose $c = A c_\alpha +B$ satisfies the equation for some
  $\rho_\shapesym$. Since this $c$ is unique, the following
  verification proves the claim.
  $$
  \begin{array}{cl}
    & D \grad^2 c - \alpha c = - \rho_\shapesym
    \\
    \therefore&D \grad^2 (A c_\alpha + B)  - \alpha (A c_\alpha +B) =
    -  \rho_\shapesym 
    \\
    \therefore &A D \grad^2 c_\alpha  - A \alpha c_\alpha -\alpha B = -  \rho_\shapesym 
    \\
    \therefore& A (D\grad^2 c_\alpha - \alpha c_\alpha) = -
    \rho_\shapesym + \alpha B
    \\
    \therefore& A( -\rho^{(\alpha)}_\shapesym) = -\rho_\shapesym + \alpha B
    \\
    \therefore& \rho_\shapesym = A(\rho^{(\alpha)}_\shapesym) + \alpha B
  \end{array}
  $$
  
  The other direction is derived similarly and the result follows.
\end{proof}

We can relate the steady-state solution of \refeq{eq:alphadisc} with
small $\alpha$ to the steady-state solution of the dynamical system in
the previous section. In fact, we now show that there are conditions
under which the two systems are similar even though the boundary
conditions are different. The key difference is that before we assumed
a constant value for $c$ at the boundary whereas now we only assume no
flow through the boundary.  \reffig{fig:rho-1D} illustrates the
correspondence in 1-D, and the following proposition makes the claim
in 2-D precise.

\begin{figure}[ht]
  \centering
  \begin{tabular}[b]{cc}
    
    \includegraphics[height=3.5in]{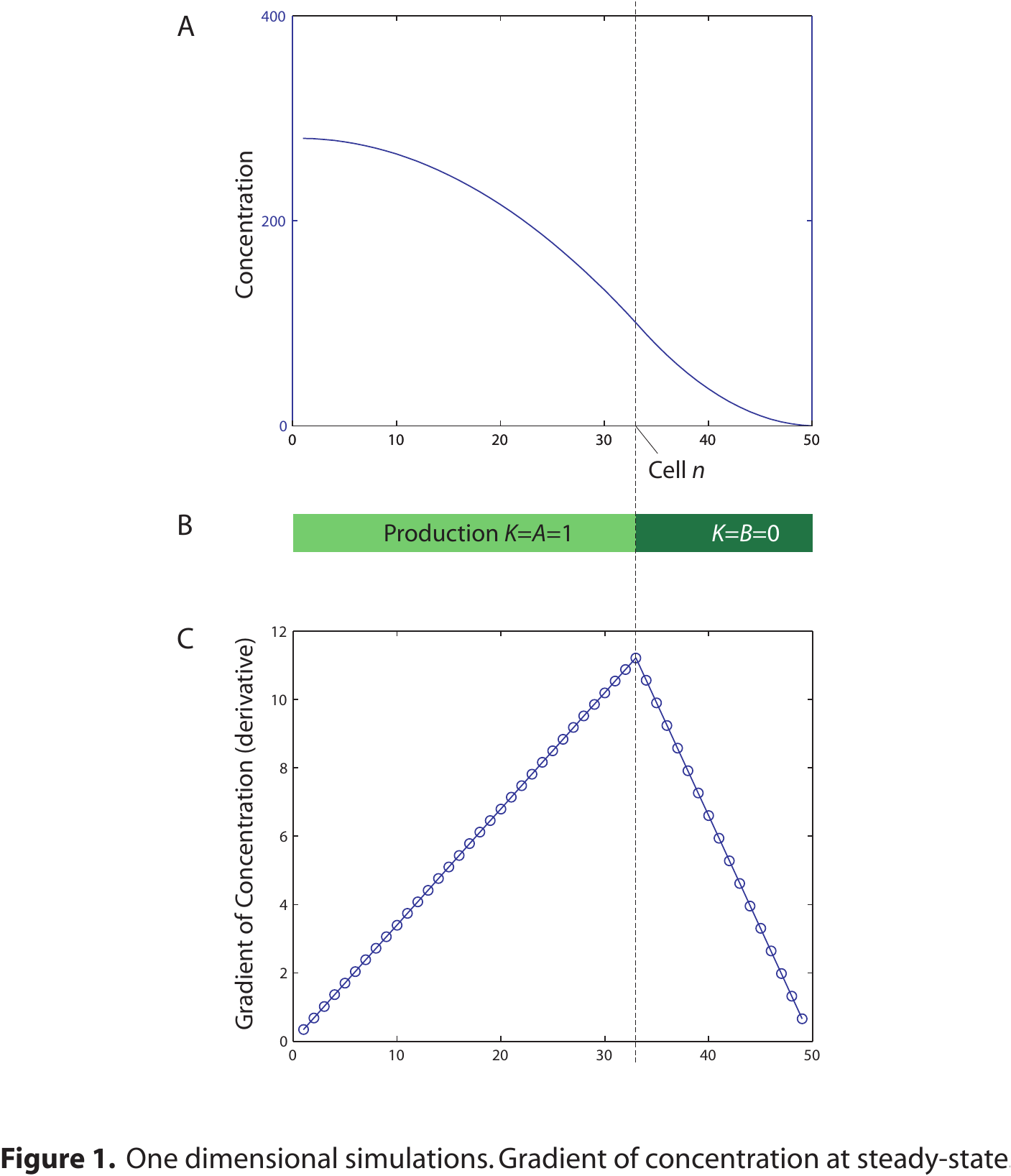} &
    
    \begin{tabular}[b]{c}
      \includegraphics[height=1.75in]{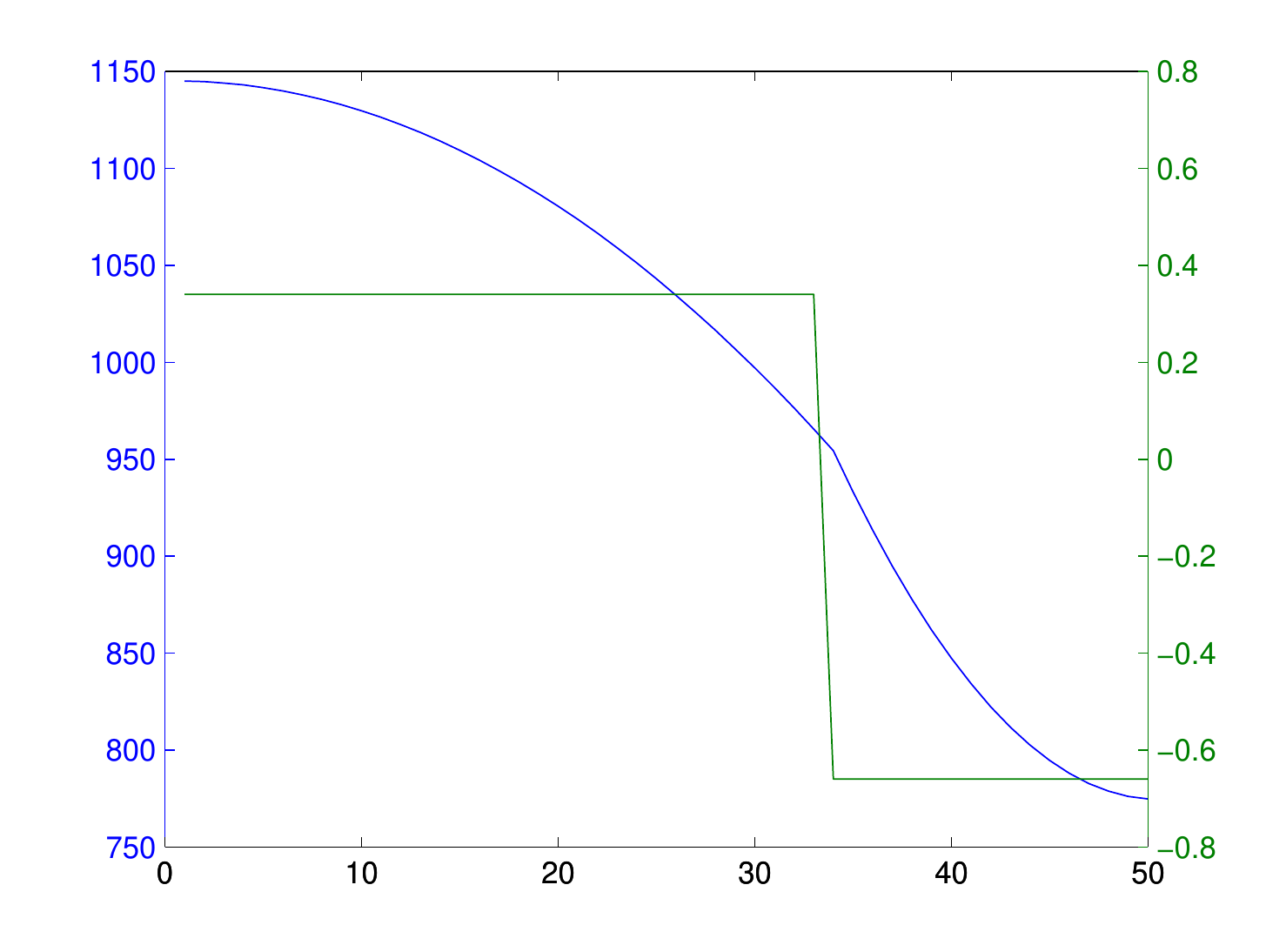} \\
      \includegraphics[height=1.75in]{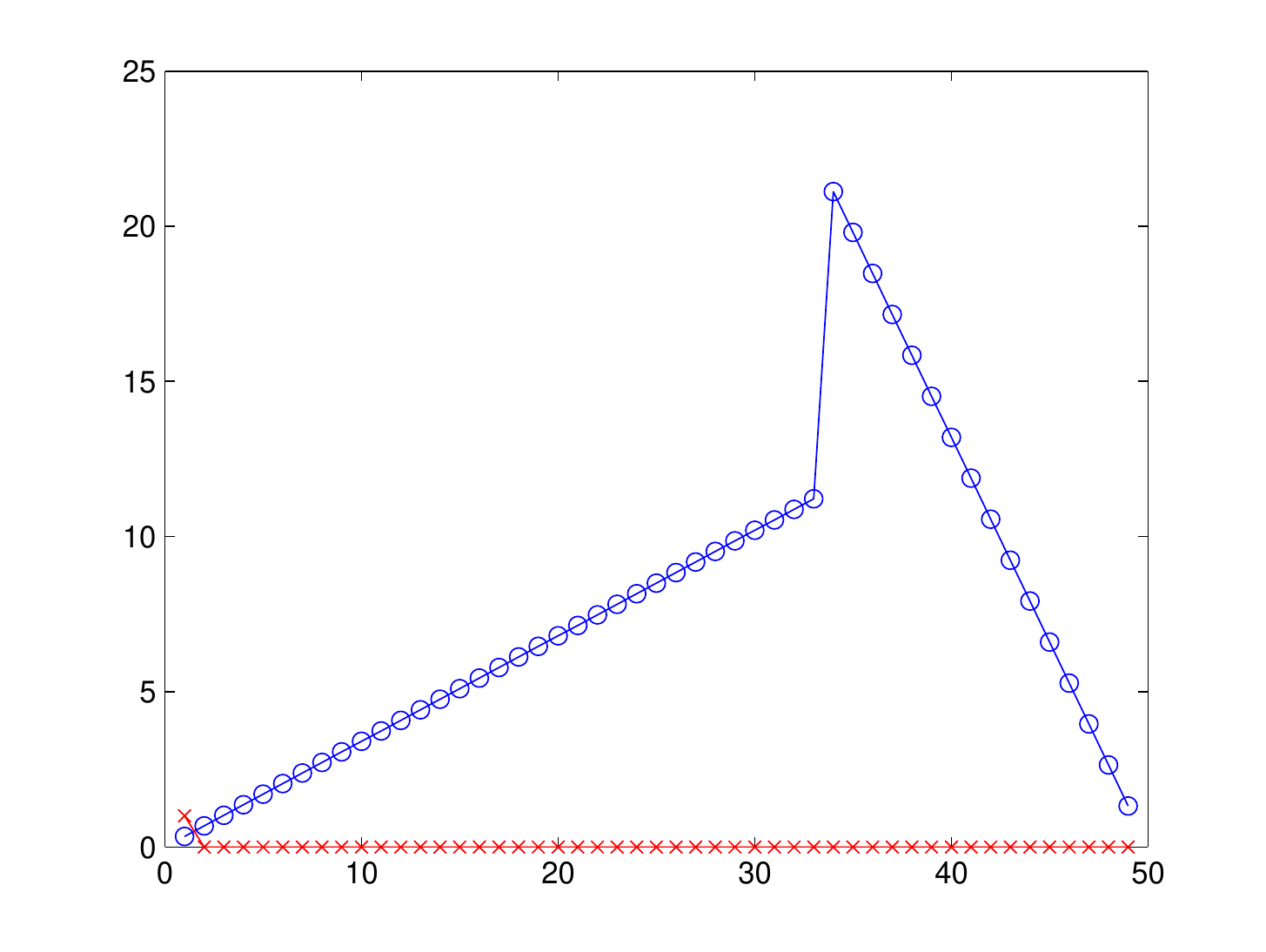} 
    \end{tabular}
    
    \\
    (a) & (b)
  \end{tabular}

  \caption[Illustration of the steady-state in a 1-D system with two
    domains of production values and small $\alpha$.]{Illustration of
    the steady-state in a 1-D system with two domains of production
    values and small $\alpha$. \FP{a} Uniform diffusion
    coefficients. \FP{b} Smaller diffusion coefficients on the
    left. {\sc top:} Concentration at equilibrium (blue curve) and
    production curve (green curve) on y-axis against cell position on
    x-axis. {\sc bottom:} Gradient of concentration (or difference in
    concentration $\Delta c$). Notice that $\Delta c$ attains a
    maximum where the production (equivalently, cell size) changes
    abruptly. Observe also that the location of this maximum has
    similar properties to the location of a sink in 1-D: e.g., the
    portion in (a) where $K=A=1$ has a parabolic concentration profile
    with a maximum where the $\Delta c$ is lowest and a minimum where
    the $\Delta c$ is largest.}
  \label{fig:rho-1D}
\end{figure}

\begin{prop}\label{prop:asin-cstbdry}
  Let $\shapesym$ be a shape with two components $\shapesym =
  \shapesym_0 \cup \shapesym_1$ such that $\shapesym_0 \cap
  \shapesym_1 = \del \shapesym_0$. Let $D_0$ and $D_1$ be the
  diffusion coefficients inside $\shapesym_0$ and $\shapesym_1$
  respectively. If $\int_{\shapesym_0} \rho_\shapesym \d{v} +
  \int_{\shapesym_1} \rho_\shapesym \d{v}= 0$ and
  $\rho_\shapesym(\shapesym_0)=K\int_{\shapesym_0}\d{v}>0$, then
  $$
  \lim_{{D_0}/{D_1} \to 0} c_\alpha = c_K
  $$
  where $c_K$ satisfies \refthm{theorem:maintheorem} for the
  shape $\shapesym_0$ by setting $c_K(\del\shapesym_0)=0$.

  \begin{proof}
    The convergence of the system derives from \refprop{prop:1}(b).
    As $D_0/D_1 \to 0$ the relative speed of diffusion in
    $\shapesym_1$ increases to infinity. Thus, the concentration over
    $\shapesym_1$ will tend to a constant and, consequently, so will
    $c(\del\shapesym_0)=c(\shapesym_0\cap\shapesym_1)$. The conditions
    of \refthm{theorem:maintheorem} are therefore satisfied and the
    claim follows.
  \end{proof}
  
\end{prop}

Note, however, that the destruction constant must be sufficiently
small in order to obtain a good correspondence. But the value of
$\alpha$ has a more important role. A strictly positive
$\alpha$ implies that there is a maximal distance beyond which
information cannot travel. Assuming that measurements that cells can
perform have a limited precision, our next result shows that the
contribution of cell A to the concentration at cell B will be
negligible whenever A is sufficiently far from B. The larger the value
of $\alpha$, the shorter this distance needs to be.

\begin{theorem}\label{thm:highalpha}
  Consider the dynamical system in \refeq{eq:alphadisc} and assume the
  conditions on the domain as in \refprop{prop:1}. Let $G(r;\sigma) =
  \exp (-r^2 / (2 \sigma^2))$ and consider the Gaussian convolution
  kernel $G_1= (\sigma^3/{2\pi}) G(\sqrt{x^2 + y^2};\sigma)$, $\sigma =
  1/\alpha$. Then
  $$
  \lim_{\alpha\to \infty} (G_1 \conv \rho_\shapesym) = c_\alpha
  $$
  where $\conv$ denotes convolution and $c_\alpha$ is as in
  \refprop{prop:1}.

  \begin{proof}
    By inspection, $c^* = G_1 \conv \rho_\shapesym$ satisfies
    \refeq{eq:alphadisc} as $\alpha \to \infty$.
    \def\S{\rho_\shapesym}

    Suppose that we have a convolution kernel $G_2$ for which $\sigma$ is a
    function of $\alpha$ and such that $\int_\Omega \alpha G_2 = 1$
    and $\sigma(\alpha) \to 0$ as $\alpha\to\infty$. Therefore, $G_2
    \conv \S \to S$ so $S - \alpha (G_2 \conv S) \to 0$. We now show
    that $G_1$ has this property, and that $\nabla^2(G_2 \conv S) \to
    0$ which shows that
    $$
    \lim_{\alpha\to\infty} D \nabla^2 c^* + \S - \alpha c^* = 0
    $$

    \noindent
    and proves the claim.

    The extrema of $G_{xx}=\del G^2 / \del x^2$ are at $r =
    \sqrt{x^2+y^2} = 0, -\sigma\sqrt{3}, \sigma\sqrt{3}$. The values
    are $G_{xx}(0) = -1/\sigma^2 G(0) = -1/\sigma^2$, and
    $G_{xx}(\sigma\sqrt{3}) = 2/\sigma^2\exp(-3/2)$. Hence, choosing
    $G_1 = \sigma^3/\sqrt{2\pi} G$ implies that $\sup \abs{\del
      G_1^2/\del x^2} = O(\sigma)$ and that $\int_\Omega G_1 = O(\sigma)$
    (because $\int_{\Reals^2} G(\sqrt{x^2 + y^2};\sigma)= 1/(\sigma^2
    {2\pi}$).  Thus, setting $\alpha = 1/\sigma$, we have that
    $\nabla^2 (G_1\conv \S) = \sigma \to 0$ and that $\int \alpha G_1
    = 1$. The claim now follows.
  \end{proof}
\end{theorem}

\begin{corol}
  Suppose a shape $\shapesym$ has two components $\shapesym =
  \shapesym_0 \cup \shapesym_1$ such that $\shapesym_0 \cap
  \shapesym_1 = \del \shapesym_0$. If $\rho_\shapesym(\shapesym_0)=0$
  and $\rho_\shapesym(\shapesym_1)=1$, then 
  $$
  \lim_{\alpha\to\infty} \frac{ \log c_\alpha(Q) - \log c_\alpha}{
    \distsym_{\shapesym_0}} = cst. >0
  $$
  where $Q\in\del\shapesym_0$ and $\distsym_{\shapesym_0}$ is the
  distance function on $\shapesym_0$.
\end{corol}

\subsection{Experimental Support of the Helmholtz Model}
\label{sec:cph-support}

The principal biological support of our Poisson model is an indirect
one: the model produces patterns that are similar to patterns observed
in nature. But it does not, for example, predict the concentrations of
auxin in any measurable fashion. By contrast, our Helmholtz model does
make such predictions and some experimental data is available.
Ljung~\etal~\cite[p.~466 and Fig.~1]{ljung-etal-2001} ``observed an
inverse correlation between leaf size and IAA concentration that was
independent of growth conditions and developmental stage.'' They
measured the proportion of hormone mass to total leaf mass $p_{IAA}$
(with units pg$\cdot$mg$^{-1}$ ) in leaves of different weight, $W$,
and obtained data that can be described well by a function $p_{IAA} =
A W^{-x}$ where $A$ is a constant and $x$ ranges between 0.72 and
0.98. This suggests that all cells may be producing auxin at the same
constant rate if the hormone is depleted proportionally
to its concentration. We reason as follows.

\begin{figure}[t]
\centering
\includegraphics{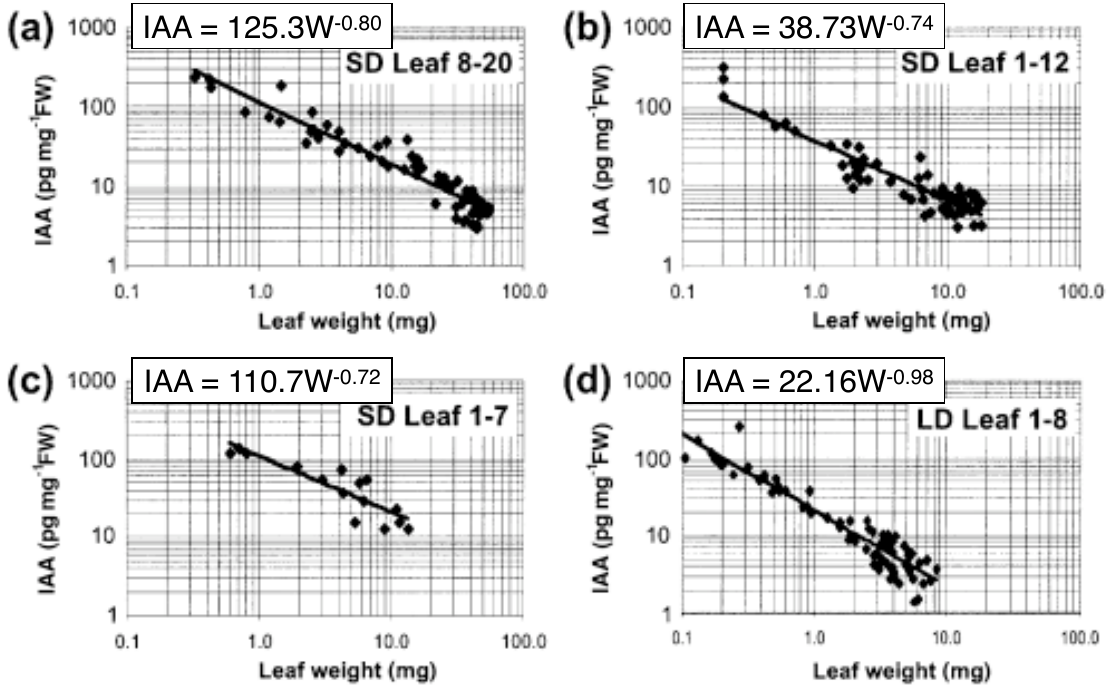}

\caption[Inverse correlation between leaf size and IAA
  concentration.]{\label{fig:leaf-iaa-powerlaw} Inverse correlation
  between leaf size and IAA concentration. After Figure 1 of
  \cite{ljung-etal-2001}. SD: short day; LD: long day. Original
  caption follows. IAA levels in Arabidopsis leaves.  The IAA
  concentration was measured in \FP{a} leaves 8-–20 from six plants
  grown for six weeks under SD \FP{b} leaves 1–-12 from six plants
  grown for 4.5 weeks under SD \FP{c} leaves 1–-7 from three plants
  grown for 3 weeks under SD and \FP{d} leaves 1-–8 from 10 plants
  grown for 16 days under LD.  The data are presented as $\log
  \log^{-1}$ plots of the IAA concentration in individual leaves
  versus leaf weight. }
\end{figure}

Let $m(i)$ denote the mass of auxin in cell $i$ and consider how this
quantity changes with time as auxin is produced, destroyed and
transported to and from neighboring cells. The production is a
constant, $K$, and the destruction, as we argued above, is
proportional to the available amounts so this rate of change can be
expressed as $m_t(i)=K - \alpha m(i) + Transport$. Assuming that the
leaf is detached, as it is prior to measurement, the transport term
only moves the hormone inside the leaf but does not contribute to
either a total increase or a total decrease. Therefore, the rate of
change of auxin mass in the whole leaf is $M_t = \sum m_t(i) = \sum K
- \alpha \sum m(i)$. The equilibrium of this system, when $M_t=0$,
describes well the state of the leaf during measurement because the
leaf is small (most leaves are less than 10 mg) and the time needed to
perform the manipulations---during which these dynamics apply
exactly---is therefore sufficiently long to shift the distribution of
the attached leaf to this equilibrium. Consequently, the constant
production hypothesis predicts a total auxin mass of around $M_{IAA} =
nK/\alpha$ for a leaf with $n$ cells. The quantity reported in the
literature, however, is an auxin-to-leaf weight ratio which we
calculate to be $p_{IAA} =
M_{IAA}/W$. Ljung~\etal~\cite{ljung-etal-2001} plot such ratios for
four groups of leaves of different sizes against leaf weight. Group
members are selected according to the order of leaf emergence
(phyllotaxis) and are organized as follows: the first group contains
samples from leaf numbers 8--20, the second from leaves 1--12, the
third from leaves 1--7, and the fourth from leaves 1--8. Leaves within
each of these ranges have a fairly similar final shape and size
(Ref. \cite{tsukaya-2002}) from which we infer that all leaves in the same
group have roughly the same final number of cells. Therefore, since
most data points are obtained after cell division has ceased and the
cell numbers have stabilized, our analysis suggests that $M_{IAA}$ is
roughly the same for all samples in the same group and that only $W$
differs. Theoretically, then, we expect the curves to be described by
a function $p_{IAA} = A W^{-x}$ with $A$ a constant and $x=1$ which is
in good agreement with the experimentally derived values of $x\approx
0.72, 0.74, 0.80, 0.98$, \reffig{fig:leaf-iaa-powerlaw}.

\subsection{Further Predictions of the Helmholtz Model} 

\begin{figure}[t]
  \centering
  \includegraphics{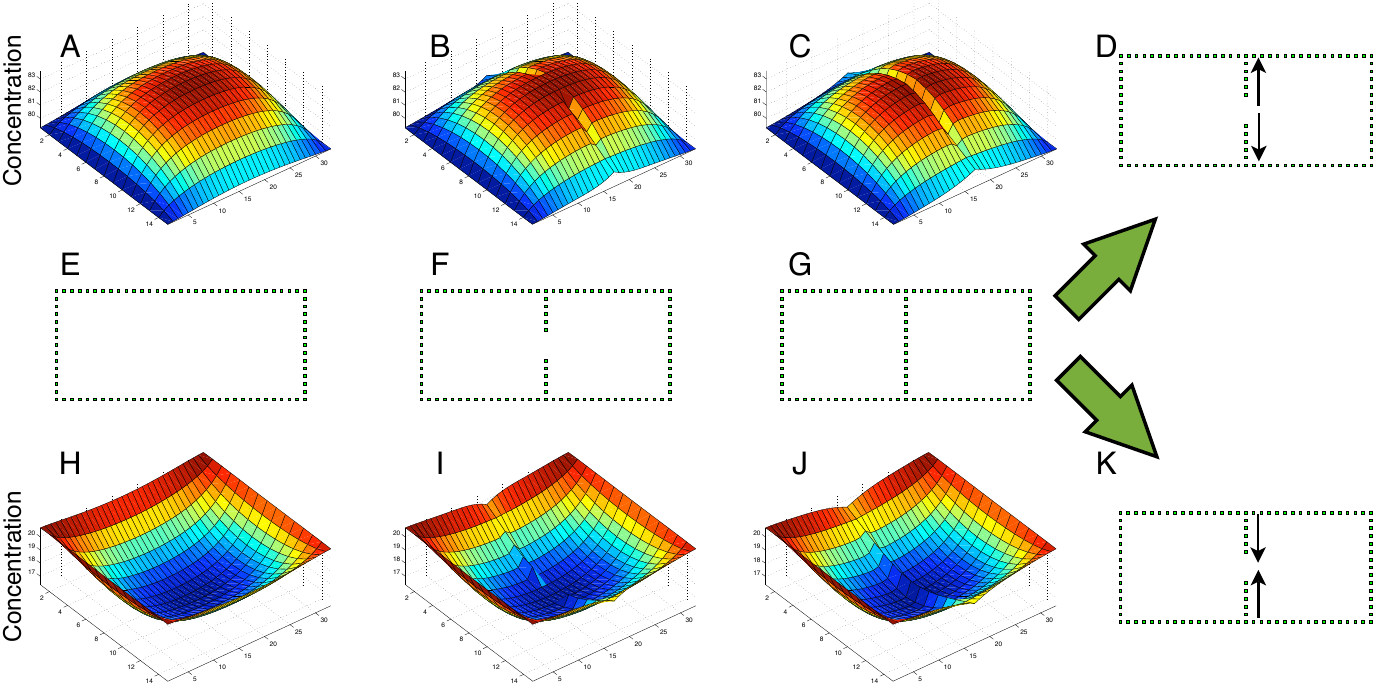}
  \caption[Our schema for the elaboration of new c-vascular strands is
    unaffected by the relative sizes of boundary cells and interior
    cells---a consistent difference is
    enough.]{\label{fig:strand-animation} Our schema for the
    elaboration of new c-vascular strands is unaffected by the
    relative sizes of boundary cells and interior cells---a consistent
    difference is enough. However, when all cells perform the
    functions CF1 and CF3 from Table~\ref{tab:helmholtz-model} the equilibrium
    concentration of $s$ is substantially different for the two types
    of configuration.  \FP{A--C} C-vascular cells are larger than
    ground cells; \FP{E--G} depiction of c-vascular cells (green
    squares) as new strands form; \FP{H--J} c-vascular cells smaller
    than ground cells. \FP{D,K} Arrows show direction of auxin
    flow. Here $\alpha=0.01$ in all simulations.}
  
\end{figure}

\subsubsection{Vein Formation}


The new (Helmholtz) formulation of the model preserves the distance
information available locally to cells under appropriate
conditions. Thus, if $\alpha$ is small and there are consistent
differences in cell size between ground cells and c-vascular cells,
then the distribution of auxin in an areole encodes size information
just as it did in the Poisson case (details in
Proposition~\ref{prop:asin-cstbdry}). For example, if ground cells are
smaller than c-vascular cells (as in
\reffig[A]{fig:strand-animation}), then the same qualitative
distribution of auxin is obtained as in \cite{pdswz:paper}.  The
program in Table~\ref{tab:helmholtz-model} then creates new strands as before.
Such relative cell sizes are observed in the early emergence of
c-vascular networks (see p.~21 in \cite{pray-1955b}), but in more
mature tissues it is the ground cells that are larger (e.g. the
procambium in Fig.~2 of \cite{nelson97:review} or the mature vein
cells compared to others in \cite[p.~234]{salisbury-ross-1992} or
\cite[p.~460]{raven81:biology-of-plants}). Our model accommodates this
second possibility as well.  If boundary cells are smaller than
interior cells (as in \reffig[I]{fig:strand-animation}), then the
hormone distribution will be inverted---higher concentration on the
boundary than in the interior---but the differences in concentration
between neighboring cells will follow the same qualitative rules as in
the first relative size configuration, albeit with an opposite sign
(details in Proposition~\ref{prop:1}). Therefore, the same program can
produce new vascular strands and the strands that it produces will be
exactly the same in both configurations.

\begin{figure}[t]
  \centering
  \includegraphics[width=\textwidth]{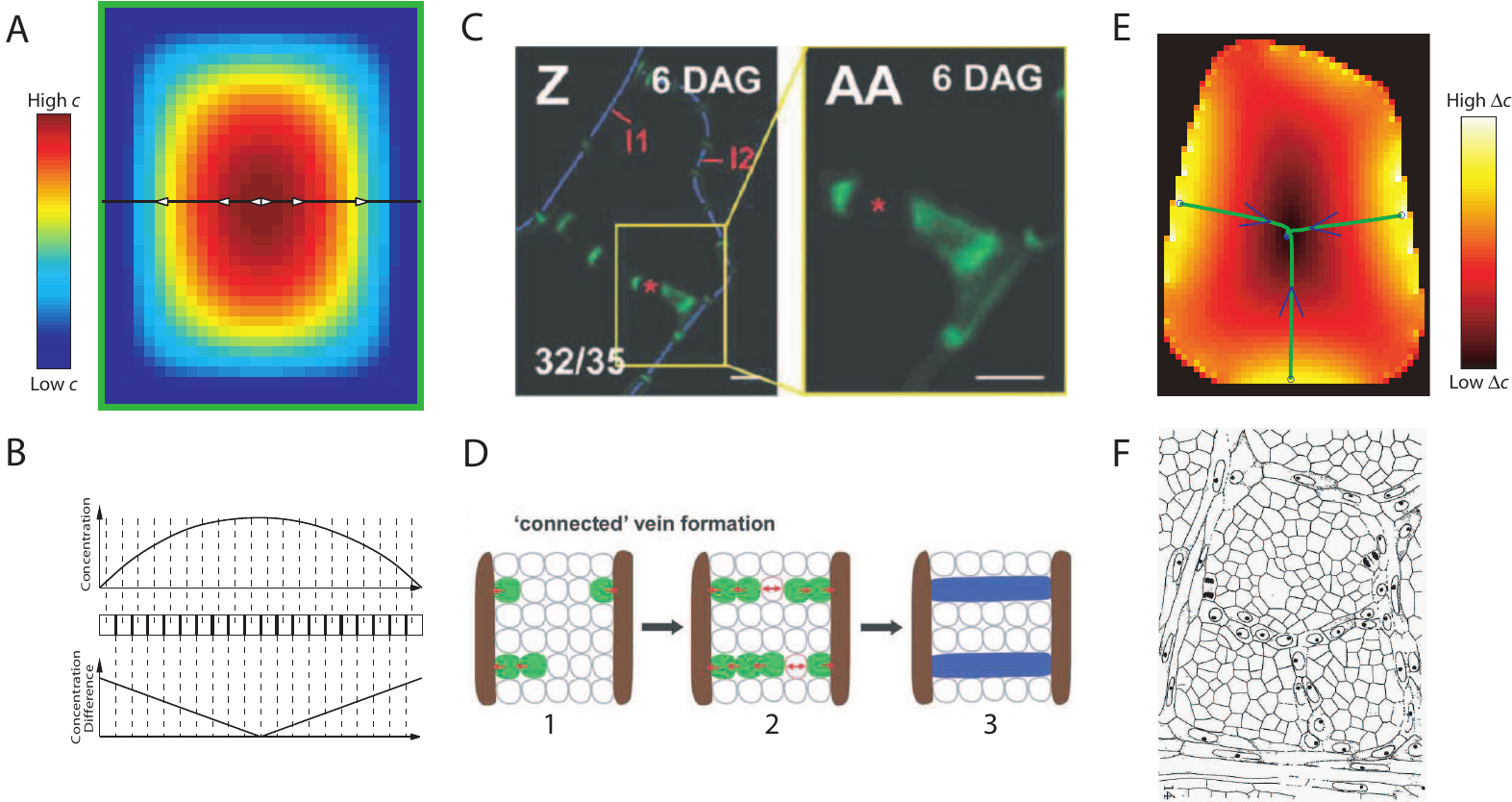}

  \caption[Illustration of the constant production hypothesis and
    cartoon model for vein formation.]{Illustration of the constant
    production hypothesis and cartoon model for vein formation.
    \FP{A} Consider a square areole as in
    \reffig[A--D]{fig:strand-animation} in which c-vascular cells are
    larger than the interior ground cells. Auxin diffuses faster
    between c-vascular cells than any other type. We show the
    equilibrium concentration distribution. Note that it is minimal
    nearest the veins and maximal at the center; i.e., {\em it varies
      with the distance to the nearest vein}.  Arrows along a
    one-dimensional cut (black line through the center) illustrate the
    flow of auxin along this line from high- to low-concentration
    pixels.  \FP{B} Concentration along the cut in \CP{A} illustrating
    maximum at center. Magnitude of gradient (concentration difference
    between cells) varies ``inversely'' and peaks at the veins with a
    value proportional to vein distance to the center. This suggests a
    {\em schema} : differentiate from ground to vascular when gradient
    magnitude is large (equivalently: when central cells are far from
    veins).  Once a cell begins to differentiate, it clears auxin more
    efficiently thereby causing adjacent cells to differentiate, until
    new veins are formed (see \CP{D} below).  \FP{C} PIN1 expression
    in {\em Arabidopsis}. Red star (*) denotes a {\em bipolar} cell.
    \FP{D} Cartoon mechanism of vein formation suggested by
    Scarpella~\etal~\cite{berleth:bipolarcell} based on measurements
    as in \CP{C}. Note that this realizes precisely the schema in
    \CP{A}.  \FP{E} Illustration of the veins formed in an areole
    according to the schema in \CP{A} and developed in
    \cite{pdswz:paper}. Note the {\em bipolar} flow at single
    cells predicted by the model and reported in
    \cite{berleth:bipolarcell}; compare with \CP{C}. Colors denote
    magnitude of gradient and arrows show the direction of vascular
    strand formation, opposite to flow. \FP{F} Cell outlines for the
    areole modeled in \CP{E}. Note the large vascular cells.  Figure
    credits: \CP{A, B, E} from \cite{pdswz:paper}; \CP{C, D} from
    Fig.~2 and Fig.~7, resp. of \cite{scarpella:genes-for-procambium};
    \CP{F} from \cite{pray-1955b}.}

  \label{fig:bipolar-cell-prediction}
\end{figure}

\reffig{fig:strand-animation} shows a simulation of new strands
created according to Table~\ref{tab:helmholtz-model} in a hypothetical areole
for each of the two combinations of relative cell sizes.  Notice that
the two new c-vascular strands connect in the middle of the
areole. They do so by draining $s$ in opposite directions so that the
cell where the two strands eventually meet does not have a well
defined polarity---it is effectively bipolar. Thus, our model predicts
multi-polar cells whenever loops of veins form. Recently at least the
bipolar case has been reported for {\em Arabidopsis}
\cite[Fig.~2]{berleth:bipolarcell},
\reffig{fig:bipolar-cell-prediction} compares our predictions to the
empirical observations.

However, our theory explains this phenomenon only under the assumption
that vascular cells are larger than ground cells. If these relative
sizes are reversed, then we should expect bipolar cells to form but
the definition needs to be revised. In effect, the relevant cells
would not have carriers facilitating export; instead, neighboring
cells from opposite sides of a bipolar cell both exhibit facilitated
transport toward the bipolar cell. Our theory predicts that such
configurations would arise if new vascular strands are created in more
mature organs, extending an existing vein as opposed to existing
procambium (e.g. a tertiary vein stemming from a primary may be a good
candidate).

\subsubsection{Auxin Distribution in {\em Arabidopsis} Roots}

Roots have a much simpler cell size distribution than leaves, and
reliable estimates are easier to obtain. There even exist detailed 3-D
models of root tips of {\em Arabidopsis}, but we shall only use 2-D
slices to test our theory in those organs. This type of data is
representative of the full 3-D organ because roots are radially
symmetric: rotating the 2-D slice about its long axis yields a good
approximation the complete root. Along the same axis, three regions
are distinguished: (1) a division zone (DZ) near the tip, (2) an
elongation zone immediately adjacent to the DZ, and (3) a maturation
or differentiation zone (\cite[p.~436]{raven81:biology-of-plants},
\cite{salisbury-ross-1992}). On average, cells are smallest in the
division zone, followed by slightly larger ones in the elongation
zone, and then by the largest cells in the maturation zone.

This schematic configuration is depicted in \reffig[A]{fig:rootsims}
and forms the setup of the first type of root simulations. Each cell
is represented by a rectangular box containing a single dot inside,
the unit of auxin production independent of cell size. Note that the
hormone distribution obtained predicted in this fashion
(\reffig[B,E]{fig:rootsims}) qualitatively agrees with the empirical
measurements (\reffig[C,D]{fig:rootsims}) reported by
Bhalerao~\etal~\cite{bhalerao-etal-2002}. Their data come from slices
of untreated plants and consist of the average concentration of IAA as
a function of distance from the root tip. The technique for measuring
auxin levels gives good concentration estimates but has poor spatial
resolution, so our comparison is restricted to the overall shape of
the curve.  The use of staining techniques, on the other hand,
promises higher resolution but only provides qualitative
information. Thus, we suspect that there is a peak of auxin
concentration near the root tip, as shown in
\reffig[C]{fig:tracedrootsims}, but do not know its height. This
observation may be sufficiently explained by the geometry of the
organ, as our first schematic simulation suggests, so we turn to a
real root specimen to test this claim.

Our next simulation uses a manually traced 2-D slice
\cite[Fig.~1A]{casimiro-etal-2001} to predict the shape of auxin
distribution inside that root. The setup is redrawn in
\reffig[A]{fig:tracedrootsims}. Notice that the simulation result (in
part {\em B}) reproduces that peak. This is a robust feature of the
model as the only external cue has been cell size: the diffusion
coefficients are all equal, the per-cell auxin production is constant.


\begin{figure}[ht]
  \centering
  \begin{center}
    \begin{minipage}[b]{0.55\linewidth}
      \includegraphics[width=.55\figwidth]{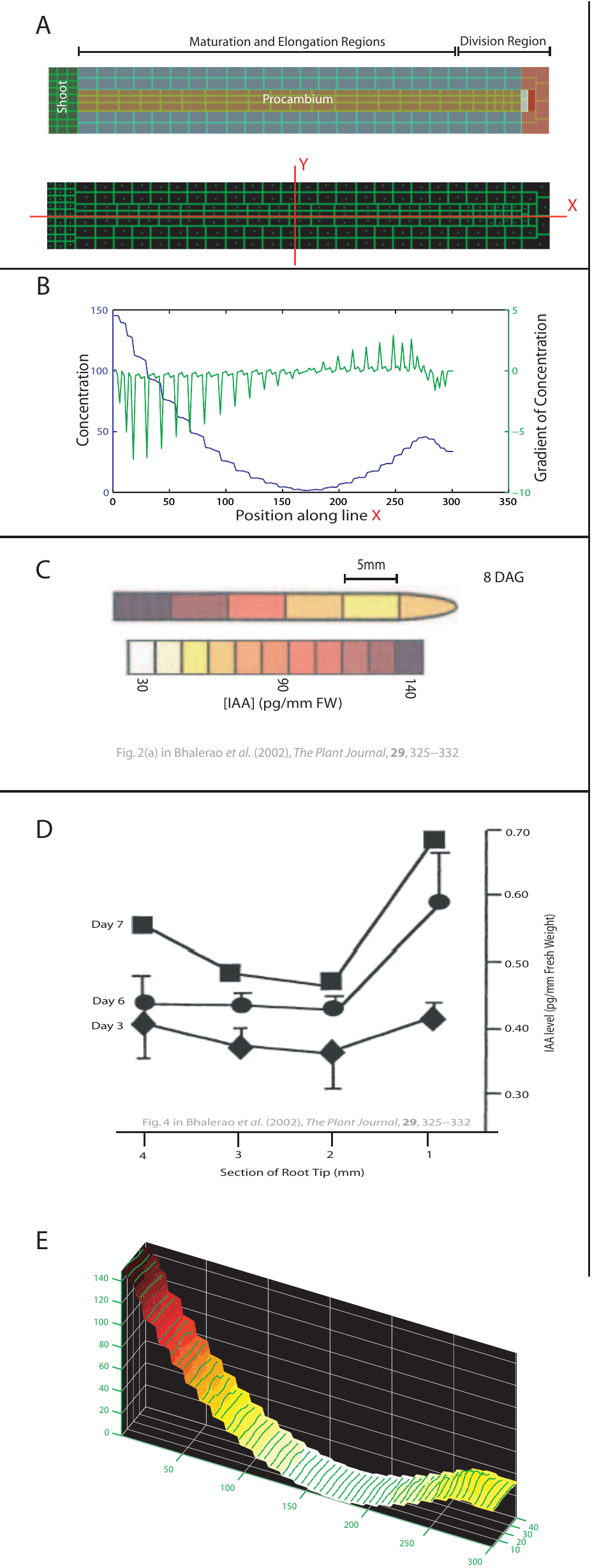}      
    \end{minipage}\hfill
    \begin{minipage}[b]{0.43\linewidth}
      \caption[Simulations in a schematic root.]{Simulations in a
        schematic root. \FP{A} Setup inspired by known regions of
        plant roots \cite[p.~434-6]{raven81:biology-of-plants}. Cells
        are rectangular boxes which increase in size from the tip
        toward the shoot. A single dot inside each cell represents one
        unit of hormone production; green boundaries (cell interfaces)
        have the same diffusion coefficients everywhere. \FP{B} Result
        of simulation. Concentration at steady-state through the
        horizontal red line in \CP{A}. \FP{C,D} Measurements of auxin
        concentration in sampled root tissues (either 5mm or 1mm
        cylinders) reported by Bhalerao~\etal~\cite[Figs. 2 and
          4]{bhalerao-etal-2002}. \FP{E} Result of simulation:
        concentration profile over the whole domain.}
      \label{fig:rootsims}
    \end{minipage}
  \end{center}
\end{figure}

\begin{figure}[ht]
  \centering
  \begin{center}
      \includegraphics[]{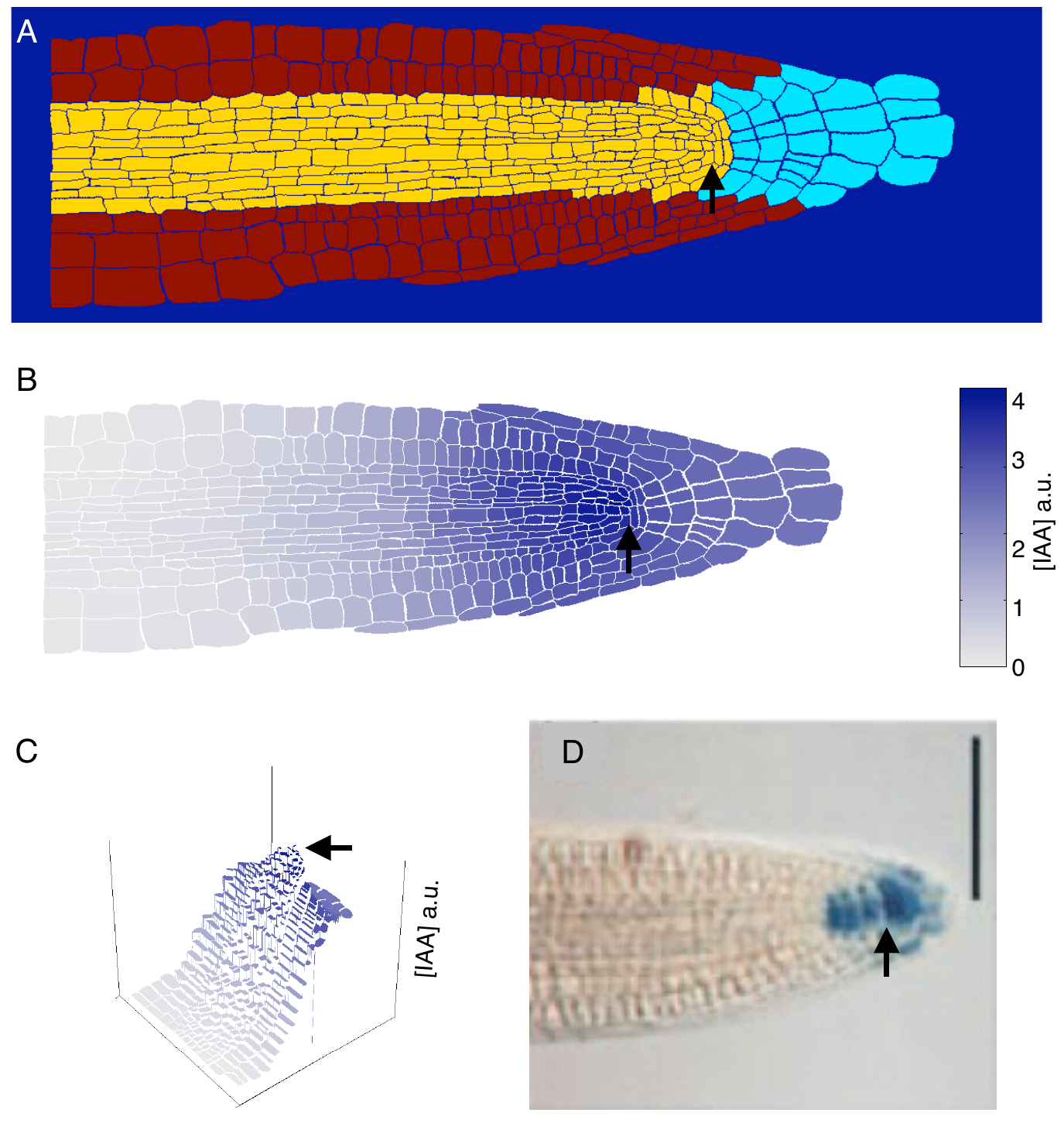}
      \caption[Root simulations from a manually traced {\em
          Arabidopsis} specimen from
        Swarup~\etal~\cite{swarup01}.]{Root simulations from a
        manually traced {\em Arabidopsis} specimen from
        Swarup~\etal~\cite[Fig.~1]{swarup01}. \FP{A} Traced root used
        for the simulation. Cell size was determined by computing the
        area; neighbor relations and interface area were obtained by
        computing the length of a shared boundary. Diffusion
        coefficients and per-cell hormone production are the same for
        all cells. \FP{B,C} Result of simulation. Note the predicted
        peak of concentration (arrows). \FP{D} Stained root from
        Casimiro~\etal~\cite[Fig.1A]{casimiro-etal-2001} showing a
        peak at the tip (arrow).  Our model predicts this peak even
        though all parameters are kept uniform.}
    \label{fig:tracedrootsims}
  \end{center}
\end{figure}

\section{Numerical Simulations}

\subsection{Background}

Here we provide background definitions and results separated in two
major categories: geometry and graph theory. The first category is
needed to prove the results about the steady-state of certain
dynamical processes, and the second category is used to prove that a
discretization of these processes is well behaved.

A {\em graph} $G=(V,E)$ is a combinatorial structure that consists of
three things: a set of vertices (or nodes) $V$, a set of edges $E$,
and an incidence relation. The edges describe which vertices are
connected and the incidence relation attributes an order to each
connection. For example, the edge $e=(i,j)$ says that node $i$ is
connected to node $j$ and that it ``starts'' at $i$ and ``ends'' at
$j$. The number of edges of which vertex $j$ is an end is called the
{\em degree} of $i$. Each edge may take a value, called a weight, and
these values can be recorded in an {\em adjacency matrix} $A$ of size
$n\times n$ where $n=|V|$. So, $A_{ij}$ is the value for the edge
$(i,j)$. Note that the matrix will be symmetric if the direction of
the edges does not matter, i.e. $A_{ij}=A_{ji}$. This is the
undirected case, which will be discussed in this paper and describes
the (Helmholtz) dynamics of auxin concentration. In the following
paper \cite{plos-paper2}, however, directed graphs will be needed and
$A_{ij}=-A_{ji}$ in some cases.

A matrix can be described by its {\em eigenvectors} and {\em
  eigenvalues}. A vector $v$ is called an eigenvector (or
characteristic vector) of the matrix $A$ iff $Av=\lambda v$ for some
number $\lambda$, which is the corresponding eigenvalue (or
characteristic value).

The {\em graph Laplacian} $L(G)$ of a graph $G$ is given by $L(G) =
D(G) - A(G)$ where $D$ is the diagonal degree matrix ($D_{ii}=$degree
of $i$) and $A$ is the adjacency matrix. The following is a standard
result (e.g. see \cite{briggs93}).

\begin{theorem} \label{thm:graph-laplacian}
  Let $G$ be a connected graph on $n$ vertices and denote by
  $\lambda_0\leq\lambda_1\leq \cdots \leq \lambda_{n-1}$ the eigenvecotrs
  of $L=L(G)$, the graph Laplacian. Then:
  \begin{enumerate}
  \item[(a)] All $\lambda_i$ are real;
  \item[(b)] $\lambda_0=0$, with eigenvector $\v{1} = [1,1,\cdots,1]^T$;
  \item[(c)] $\lambda_1>0$; and
  \item[(d)] the eigenvectors of $L$ span $\Reals^n$.
  \end{enumerate}
\end{theorem}

A matrix $A$ is called {\em positive definite} if $v^t A v > 0$ (where
$v^t$ is the transpose of $v$) for all non-zero vectors $v$, and it is
{\em semi-definite} if some non-zero vector $w$ exists such that $v^t
A v = 0$. Note that $L(G)$ is positive semi-definite.

The {\em determinant} of an $n\times n$ matrix $A$ is given by the following
formula:
$$
det(A) = \sum_{\sigma \in S_n} sgn(\sigma) \prod_{i=1}^n A_{i,\sigma(i)} 
$$

\noindent
where $\sigma$ is a permutation on $n$ elements and $sgn(\sigma)$ is
the sign of the permutation: positive if the permutation can be
produced by an even number of element exchanges and negative
otherwise. A matrix $A$ is invertible -- i.e. $A^{-1}$ exists such
that $A^{-1}A=Id$ -- iff $det(A) \neq 0$. The determinant is equal to
the product of eigenvalues, $det(A)=\prod_i \lambda_i$, so $L(G)$ is
not invertible.

\subsection{Discrete Simulations}
\label{sec:discrete-sims}

In this section we prove that appropriate discretizations of the
continuous equations may be solved numerically. Specifically, we show
that \refeq{eq:goveq} and \refeq{eq:production-destruction} converge
to a unique solution as $t\to\infty$ and show how to obtain this
solution without iteration. We shall not assume anything about the
dimensionality of the space which contains the cells, only that they
are connected.

Suppose there are $n$ cells in a conglomerate $\shapesym$ where each
cell shares an interface with at least one other cell and that the
conglomerate is connected in this fashion. Labeling each cell with a
number from 1 to $n$, suppose that each cell $i$ contains a hormone at
concentration $c(i)$. The interfaces allow this hormone to diffuse
following Fick's law, so assuming the diffusion constant through an
interface between cell $i$ and cell $j$ is $D_{ij}$, then the flow
into cell $i$ through each existing interface is given by $D_{ij}
(c(j)-c(i))$. Thus, the diffusion of the hormone through the
conglomerate may be described by a matrix $-L$ applied to the vector
$\v{c}$ of concentrations, where
$$
-L_{ij} = \left\{
  \begin{array}{l@{\hspace{1em}}r}
    D_{ij} & i\neq j\\
    -\sum_{j\neq i} D_{ij} & i=j
  \end{array}
  \right.
  $$
  Notice that $L$ is the graph Laplacian for the graph where each cell
  is a node and the edge weights are the diffusion coefficients
  $D_{ij}$.
  
  Now suppose that the hormone concentration in cell $i$ is somehow
  maintained at a fixed level $c(i)=c_f$; refer to such a cell as a
  \defem{sink}. Then the diffusion process in or out of cell $i$ has
  no effect on its concentration $c(i)$, but neighboring cells $c(j)$
  will be affected by $c(i)$. Hence, the new matrix $M_{ij}$
  describing the diffusion process looks exactly like $-L$ except for
  the rows corresponding to sinks; if $i$ is a sink then $M_{ii}=1$
  and $M_{ij}=0$ for $j\neq i$. This means that $det(M)=det(M^{ns})$,
  and that $M$ is no longer symmetric, because any neighbor $j$ of $i$
  which is not a sink will induce $M_{ji}=D_{ji}$. However, the
  sub-matrix $M^{ns}$ of $M$ with rows and columns corresponding to all
  cells which are not sinks is symmetric. In fact, the next lemma
  follows immediately.
  
  \begin{lemma}
    Let $\shapesym$ be a conglomerate of $n = n_s + n_r$ cells of
    which $n_s$ are sinks. Label the sinks 1 to $n_s$ and the rest
    with $n_s+1$ to $n$. Let $M^{ns}=M(n_s+1:n, n_s+1:n)$ be the
    $n_r\times n_r$ sub-matrix and $-L$ be the graph Laplacian of the
    sub-graph $G^{(ns)}$ of non-sink cells. Let $C_{ij}$ by an
    $n_r\times n_r$ diagonal matrix such that for each cell $i$ in
    $G^{(ns)}$ neighboring sinks $s_j$ we have $C_{ii} = \sum_{s_j}
    D_{i s_j}$. Then, for a suitable labeling,
    $$
    M^{ns}= -L - C ~.
    $$
  \end{lemma}

  If we also assume that the hormone is destroyed in each non-sink
  cell $i$ at a rate proportional to the concentration in cell $i$,
  then the sub-matrix becomes
  \begin{equation}
    \label{eq:Mns}
      M^{ns}= -L - C - \alpha I~
  \end{equation}
  where $I$ is the $n_r\times n_r$ identity matrix. Observe that
  $M$ is invertible if and only if $M^{ns}$ is invertible because
  $det(M)=det(M^{ns})$.

  \begin{lemma} \label{lemma:Minvertible}
    If either $\alpha>0$ or $\alpha\geq 0$ and there is at least one
    sink cell, then the matrix $M$ is invertible and has strictly
    negative eigenvalues.
  \end{lemma}

\begin{proof}
  As observed above, it suffices to show that $det(M^{ns})\neq 0$. We
  shall use the following result and apply it to $-M^{ns}$ from
  \refeq{eq:Mns}.
  
  \begin{lemma}  \label{lemma:posdef}
    Let $D=\st{d_{ij}}$ be an $n\times n$ diagonal matrix with
    $d_{ii}\geq 0$ with at least one $d_{ii}>0$. Let $L$ be the graph
    Laplacian of a connected graph on $n$ vertices. Then $M=L+D$ is
    symmetric positive definite. It follows that $M$ is invertible and
    has strictly positive eigenvalues.
  \end{lemma}
  
  \begin{proof}
    From \refthm{thm:graph-laplacian} we see that $L$ is positive
    semi-definite. So, if $\v{v}\in\Reals^n$ and $\v{v}\neq 0$, then
    $\v{v}^T L \v{v} \geq 0$ with equality reached only for $\v{v}=s\v{1}$
    where $s\neq 0$.  Similarly, $D$ is positive semi-definite because
    $\v{v}^T D \v{v} = \sum_i d_{ii} v_i^2$ and all terms are non-negative.
    Further, $(s\v{1})^T D (s\v{1}) = \sum_i d_{ii}s^2 > 0$ since
    $s^2>0$ and at least one entry in $D$ is strictly positive.
    Finally, if $\v{v}\in\Reals^n$ and $\v{v}\neq 0$, then
    $$
    \v{v}^T M \v{v} = \v{v}^T (L+D) \v{v} = \v{v}^T L \v{v} +
    \v{v}^T D \v{v} > 0 ~.
    $$
    This finishes the proof of \reflemma{lemma:posdef}.
  \end{proof}
  
  Now, let $D = C + \alpha I$. If there is at least one sink, then at
  least one $C_{ii}>0$; if $\alpha>0$ then all $D_{ii}>0$. Hence
  \reflemma{lemma:posdef} applies and completes the proof of
  \reflemma{lemma:Minvertible}.
\end{proof}

To complete the discretization of the continuous process, suppose that
$\rho_\shapesym(i) = K/S(i)$, with units
($mass*time^{-1}*volume^{-1}$), is the rate at which cell $i$ produces
the hormone. Setting $S(i)$ to be the size of cell $i$, we obtain the
following discrete dynamics:
\begin{equation}
  \label{eq:discrete-process}
  \v{c}^{(t+\dt)} = \v{c}^{(t)} + \dt \underbrace{ \paren{ M
      \v{c}^{(t)} + \rho_\shapesym }}_{dc/dt }
\end{equation}
where $\dt$ is the time step. We wish to show that, given a small
enough $\dt$, $\v{c}$ will converge to a unique value.

Observe that the update rule in \refeq{eq:discrete-process} may
represented as $\tilde \v{c}^{(t+\dt)} = U \tilde \v{c}^{(t)} $ by
writing
$$
\underbrace{
\paren{
  \begin{array}{c}
    c^{(t+\dt)}(1) \\ c^{(t+\dt)}(2) \\ \vdots \\ c^{(t+\dt)}(n) \\ 1
  \end{array}
}
}_{\tilde \v{c}^{(t+\dt)}}
=
\underbrace{
\left[
  \begin{array}{cc}
    \left[
      \begin{array}{ccccc}
        \\
        & (\dt) M + I&
        \\ \\ \\
      \end{array}
    \right] &
    \begin{array}{c}
      \dt {\rho_\shapesym}(1) \\ \dt {\rho_\shapesym}(2) \\
      \vdots \\ \dt {\rho_\shapesym}(n)
    \end{array}
   \\
  0 & 1
\end{array}
\right]
}_{U}
\underbrace{
\paren{
  \begin{array}{c}
    c^{(t)}(1) \\ c^{(t)}(2) \\ \vdots \\ c^{(t)}(n) \\ 1
  \end{array}
}
}_{\tilde \v{c}^{(t)}}
$$
where $I$ is the $n\times n$ identity matrix.
The convergence of the process now reduces to showing that $U^k
\tilde \v{c}^{(0)}$ converges as $k\to\infty$. Notice that $U$ has the
block form of the following three matrices
$$
U_a = 
\left[
  \begin{array}{cc}
    A & \v{v}_a \\
    \v{0} & 1
  \end{array}
\right], \quad
U_b = 
\left[
  \begin{array}{cc}
    B & \v{v}_b \\
    \v{0} & 1
  \end{array}
\right], \quad
U_c = 
\left[
  \begin{array}{cc}
    C & \v{v}_c \\
    \v{0} & 1
  \end{array}
\right]
$$
where $A$, $B$ and $C$ are $n\times n$ matrices; $\v{v}_a$,
$\v{v}_b$ and $\v{v}_c$ are $n\times 1$ vectors; and $\v{0}$ is a
$1\times n$ vector. The product of two such matrices preserves the
block form; e.g. $U_{c}=U_a U_b$ by setting $C
= A B$ and $\v{v}_{c} = A \v{v}_b + \v{v}_a$. Therefore, by induction
on $k$, the blocks of $U_c=U_a^k$ must be $C=A^k$ and $\v{v}_c =
\sum_{i=0}^{n-1} A^i \v{v}_a$.

\begin{theorem}
  Suppose that each cell $i$ has size $S(i)$, produces a hormone at a
  rate $\rho_\shapesym(i)$ and destroys the hormone at a rate $\alpha
  c(i)$ with $\alpha\geq 0$. If $\alpha>0$ or there is at least one
  sink cell, then for a sufficiently small $\dt$ the discrete process
  in \refeq{eq:discrete-process} will converge to $\v{c}^* = -M^{-1}
  (\rho_\shapesym/S)$.
\end{theorem}

\begin{proof}
  The conditions of \reflemma{lemma:Minvertible} apply, so $M$ has
  strictly negative eigenvalues and $M^{-1}$ exists. Choose $0<\dt <
  1/\lambda$ where $\lambda$ is the largest (in absolute value)
  eigenvalue of $M$. Thus $A=\dt M + I$ will have eigenvalues
  $0<\abs{\lambda_i} <1$; it follows that $\lim_{k=\infty} A^k = 0$.
  Now, from the above, $\v{v}_c = \sum_{i=0}^{k-1} A^i \v{v}_a =
  (A^k-I)(A-I)^{-1} \v{v}_a = (A^k-I)(\dt M)^{-1} \v{v}_a $ and the
  claim follows.
\end{proof}

This result demonstrates that an iterative process will indeed
converge---assuming perfect arithmetic operations---but it also shows
that the equilibrium can be computed much more efficiently. It
suffices to solve the linear system $\v{c}^* = -M^{-1}
(\rho_\shapesym/S)$. The system is well behaved numerically whenever
$\alpha$ is sufficiently large, because the condition number of this
matrix is roughly equal to the largest degree of the graph, times $D$
divided by $\alpha$; see Dahlquist and
Bj{\"o}rck~\cite{dahlquist-bjork-1974} for a discussion of matrix
condition numbers.

\subsection{Geometric Domain Definition}
\label{sec:domain-definition}

In this section we outline how the domains---representing leaves,
roots, etc.---are defined geometrically and then converted into the
graph representation discussed in the previous section. Ultimately,
the geometry of the domains should correspond to and be comparable to
the geometry of real plant tissues. Thus, we define the domains by
manipulating images of those tissues. Both the cell size and the cell
neighbors (i.e. the topology of the graph) are computed from an image.

In this paper we adopted a pixel-based approach whereby the organ is
drawn as a digital image and the color of each pixel encodes some
information: whether the pixel is part of the domain or not, the value
of the production function $\rho$, whether the pixel is a sink or part
of the vein pattern, etc. The natural connectivity of pixels on a
square grid---four or eight neighbors---then defines the topology of
the graph. The final diffusion matrix is built after defining the
diffusion constants for each pair of pixel colors.

But this representation also allows us to define a cell by using
multiple pixels.  \reffig{fig:rootsims} shows an example in which a
cell consists of several black pixels---representing the
interior---surrounded by green pixels---representing the cell
walls. None produce auxin except for a single dark-green pixel in the
middle of the interior pixels. Thus, each cell produces auxin at the
same rate (the rate of the center pixel) but cells may have different
sizes. Moreover, the diffusion coefficient in the interior of the cell
may be different from the diffusion coefficient through the cell
wall. Our usual assumption is that the interior diffusion coefficient
is much larger.

\section{Conclusion}


We have developed the foundations for a theory of how global
information about shape is related to the distance transform, and how
several of the essential properties of this distance transform can be
computed by a simple reaction-diffusion equation. The model has its
roots in our earlier Constant Production Hypothesis, and is based on a
computational abstraction that all cells behave according to the same
rules. Most importantly, it provides a mechanism that illustrates how
``hot spots'' of concentration can develop from structural conditions
rather than differential production induced by an explicit
developmental program.

The explicit assumption about hormone depletion---the Proportional
Destruction Hypothesis---greatly increased the scope of our earlier
model \cite{pdswz:paper}. We showed that there are at least two
additional ways in which a distance map becomes locally available to a
group of cells, and that testable predictions ensue. And although the
available data is insufficient to compare the predictions to the
actual numbers, the qualitative trend is accurately reproduced and
explicit measurements are suggested by updated model. The analysis and
assumptions of \refsec{sec:cph-support}, in effect, describe further
experiments to test the theory.

The simulations in this paper, which involved detailed anatomical
considerations, show the power of such calculations. That an auxin
concentration peak emerged properly near the root tip illustrates
their role is sufficency rather than necessity.

Nevertheless, our model is still too abstract to be deemed
biological. In particular, it is well known that diffusion is not the
only transport mechanism responsible for auxin flow. Active, or at
least facilitated, transport carriers are known to exist, which our
current formulation does not consider. That is a topic of our
companion paper. For now, we remark that the Fickian transport and
reaction diffusion equation developed here can provide an abstraction
in such a manner that its main properties hold when more detailed
facilitated transport is taken into account.

\singlespacing
\bibliographystyle{abbrv}    
\bibliography{paper}

\end{document}